\documentclass[10pt,journal]{IEEEtran}

\usepackage{times}
\usepackage{algorithmic}
\usepackage{amsmath}
\usepackage{amssymb}
\usepackage{amsthm}
\usepackage{color}
\usepackage{colortbl}
\usepackage{cite}
\usepackage{verbatim}
\usepackage[pdftex]{graphicx}
\usepackage{caption}
\usepackage{subcaption}
\usepackage{url}
\usepackage{pbox}

\theoremstyle{plain}
\newtheorem{theorem}{Theorem}
\newtheorem{proposition}{Proposition}
\newtheorem{lemma}{Lemma}
\newtheorem{corollary}{Corollary}
\newtheorem{remark}{Remark}

\theoremstyle{definition}
\newtheorem{definition}{Definition}
\newtheorem{example}{Example}

\theoremstyle{remark}

\newcommand{\beq}{\begin{eqnarray}}
\newcommand{\eeq}{\end{eqnarray}}

\newcommand{\field}[1]{\mathbb{#1}}

\newcommand{\F}{\field{F}}

\newfont{\bbb}{msbm10 scaled 500}

\newfont{\bb}{msbm10 scaled 1100}

\newcommand{\FF}{\mbox{\bb F}}

\newcommand{\bv}{{\bf b}}
\newcommand{\cv}{{\bf c}}

\newcommand{\mv}{{\bf m}}

\newcommand{\xv}{{\bf x}}
\newcommand{\yv}{{\bf y}}

\newcommand{\Ac}{{\cal A}}

\newcommand{\Cc}{{\cal C}}

\newcommand{\Ec}{{\cal E}}

\newcommand{\Gc}{{\cal G}}
\newcommand{\Hc}{{\cal H}}
\newcommand{\Ic}{{\cal I}}

\newcommand{\Lc}{{\cal L}}

\newcommand{\Rc}{{\cal R}}
\newcommand{\Sc}{{\cal S}}
\newcommand{\Tc}{{\cal T}}
\newcommand{\Uc}{{\cal U}}

\newcommand{\Vc}{{\cal V}}

\definecolor{OXO-emph}{RGB}{153,0,0}

\newcommand\ceilb[1]{\left\lceil #1 \right\rceil}
\newcommand\floorb[1]{\left\lfloor #1 \right\rfloor}
\newcommand{\algrule}[1][.2pt]{\par\vskip.5\baselineskip\hrule height #1\par\vskip.5\baselineskip}
\DeclareMathAlphabet{\mathpzc}{OT1}{pzc}{m}{it}

\newcommand\nc\newcommand
\nc\bfa{{\boldsymbol a}}\nc\bfA{{\boldsymbol A}}\nc\cA{{\mathcal A}}
\nc\bfb{{\boldsymbol b}}\nc\bfB{{\boldsymbol B}}\nc\cB{{\mathcal B}}
\nc\bfc{{\boldsymbol c}}\nc\bfC{{\boldsymbol C}}\nc\cC{{\mathcal C}}
\nc\sC{{\mathscr C}}
\nc\bfd{{\boldsymbol d}}\nc\bfD{{\boldsymbol D}}\nc\cD{{\mathcal D}}
\nc\bfe{{\boldsymbol e}}\nc\bfE{{\boldsymbol E}}\nc\cE{{\mathcal E}}
\nc\bff{{\boldsymbol f}}\nc\bfF{{\boldsymbol F}}\nc\cF{{\mathcal F}}
\nc\bfg{{\boldsymbol g}}\nc\bfG{{\boldsymbol G}}\nc\cG{{\mathcal G}}
\nc\bfh{{\boldsymbol h}}\nc\bfH{{\boldsymbol H}}\nc\cH{{\mathcal H}}
\nc\bfi{{\boldsymbol i}}\nc\bfI{{\boldsymbol I}}\nc\cI{{\mathcal I}}
\nc\bfj{{\boldsymbol j}}\nc\bfJ{{\boldsymbol J}}\nc\cJ{{\mathcal J}}
\nc\bfk{{\boldsymbol k}}\nc\bfK{{\boldsymbol K}}\nc\cK{{\mathcal K}}
\nc\bfl{{\boldsymbol l}}\nc\bfL{{\boldsymbol L}}\nc\cL{{\mathcal L}}
\nc\bfm{{\boldsymbol m}}\nc\bfM{{\boldsymbol M}}\nc\sM{{\mathscr M}}
\nc\bfn{{\boldsymbol n}}\nc\bfN{{\boldsymbol N}}\nc\cN{{\mathcal N}}
\nc\bfo{{\boldsymbol o}}\nc\bfO{{\boldsymbol O}}\nc\cO{{\mathcal O}}
\nc\bfp{{\boldsymbol p}}\nc\bfP{{\boldsymbol P}}\nc\cP{{\mathcal P}}
\nc\bfq{{\boldsymbol q}}\nc\bfQ{{\boldsymbol Q}}\nc\cQ{{\mathcal Q}}
\nc\bfr{{\boldsymbol r}}\nc\bfR{{\boldsymbol R}}\nc\cR{{\mathcal R}}
\nc\bfs{{\boldsymbol s}}\nc\bfS{{\boldsymbol S}}\nc\cS{{\mathcal S}}
\nc\bft{{\boldsymbol t}}\nc\bfT{{\boldsymbol T}}\nc\cT{{\mathcal T}}
\nc\bfu{{\boldsymbol u}}\nc\bfU{{\boldsymbol U}}\nc\cU{{\mathcal U}}
\nc\bfv{{\boldsymbol v}}\nc\bfV{{\boldsymbol V}}\nc\cV{{\mathcal V}}
\nc\bfw{{\boldsymbol w}}\nc\bfW{{\boldsymbol W}}\nc\cW{{\mathcal W}}
\nc\bfx{{\boldsymbol x}}\nc\bfX{{\boldsymbol X}}\nc\cX{{\mathcal X}}
\nc\bfy{{\boldsymbol y}}\nc\bfY{{\boldsymbol Y}}\nc\cY{{\mathcal Y}}
\nc\bfz{{\boldsymbol z}}\nc\bfZ{{\boldsymbol Z}}\nc\cZ{{\mathcal Z}}

\newcommand{\remove}[1]{}

\newcommand\ff{{\mathbb F}}

\begin{document}

\sloppy

\title{Cooperative Local Repair in Distributed Storage}

\author{Ankit~Singh~Rawat, Arya Mazumdar, and~Sriram~Vishwanath
\thanks{A.~S.~Rawat and S.~Vishwanath are with the Laboratory of Informatics, Networks and Communications, Department of Electrical and Computer Engineering, The University of Texas at Austin, Austin, TX 78751 USA. E-mail: ankitsr@utexas.edu, sriram@austin.utexas.edu.}
\thanks{A.~Mazumdar is with the Department of Electrical and Computer Engineering, University of Minnesota -- Twin Cities, Minneapolis, MN 55455 USA. E-mail: arya@umn.edu.}
\thanks{This paper was presented in parts at the 48th Annual Conference on Information Sciences and Systems (CISS), March 2014.}}

\maketitle

\begin{abstract} 
Erasure-correcting codes, that support {\em local repair} of codeword symbols, have attracted 
substantial attention recently  for their application in distributed storage systems.
This paper investigates a generalization of the usual locally repairable codes. In particular, this paper studies a class of  codes with the following property: any small set of  codeword symbols can be reconstructed (repaired) from a small number of other symbols. This is referred to as {\em cooperative local repair}. The main contribution of this paper is bounds on the trade-off of the minimum distance and the dimension of such codes, as well as explicit constructions of families of codes that enable cooperative local repair.  Some other results regarding cooperative local repair are also presented, including an analysis for the well-known Hadamard/Simplex codes.
\end{abstract}






\section{Introduction}

In this paper we explore a new class of codes that enable efficient recovery from the failure of multiple code symbols. In particular, we study codes with $(r, \ell)$-cooperative locality which allow for any $\ell$ failed code symbols to be recovered by contacting at most $r$ other intact code symbols. Our study of such codes is motivated by their application in distributed storage systems {\em a.k.a.} cloud storage, where information is stored over a network of storage nodes (disks). In order to protect the stored information against inevitable node (disk) failures, a distributed storage system encodes the information using an erasure-correcting code. The code symbols from the obtained codeword are then stored on the nodes in the system. {Each node stores one code symbol from the codeword.}

The task of recovering the code symbols stored on failed nodes with the help of the code symbols stored on intact nodes is referred to as {\em code repair} or {\em node repair}~\cite{dimakis}. An erasure-correcting code with an efficient code repair process helps quickly restore the  state after node failures. This consequently enables seamless operation of the system for a long time period. Recently, multiple classes of erasure-correcting codes have been proposed that optimize the code repair process with respect to various performance metrics. In particular, the codes that minimize {\em repair-bandwidth}, {\em i.e.}, the number of bits communicated during repair of a single node, are studied in \cite{dimakis,RSK11,zigzag13,PapDimCad_Hadamard} and references therein. The codes that enable small disk-I/O during the repair process are studied in \cite{khan2011search,zigzag13}. Another family of erasure codes that focus on small {\em locality}, {\em i.e.}, enabling repair of a single failed code symbol by contacting a small number of other code symbols, are presented in \cite{Gopalan12, PapDim12, oggier_hom, RKSV12, KPLK12, TamoBarg}. 

A code is said to have {\em all-symbol locality} $r$ if every code symbol is a function of at most $r$ other code symbols. This ensures {\em local repair} of each code symbol by contacting at most $r$ other code symbols. In this paper we generalize the notion of codes with all-symbol locality to codes with $(r, \ell)$-cooperative locality: any set of $\ell$ code symbols are functions of at most $r$ other code symbols. This allows for {\em cooperative local repair} of code symbols,  where any group of $\ell$ failed code symbols is repaired by contacting at most $r$ other code symbols. 

The ability to perform code repairs involving more than one failure is a desirable feature in most of the distributed storage systems that can experience multiple simultaneous failures~\cite{ford}. Moreover, this property also allows for deliberately delaying code repairs when system resources need to be freed to support other system objectives, {\em e.g.}, queries (accesses) to the stored information by clients. Here, we note that the approach of cooperative code repair has been previously explored in the context of repair-bandwidth efficient codes in \cite{Kermarrec, ShumHu} and references therein. 

In this paper we address two important issues regarding codes with $(r, \ell)$-cooperative locality: 1) obtaining trade-offs among minimum distance, dimension (rate), and locality parameters $(r, \ell)$ for such code; and 2) presenting explicit constructions for codes with $(r, \ell)$-cooperative locality that are close to the obtained trade-offs. Towards designing codes with $(r, \ell)$-cooperative locality, we mainly focus on codes with maximum possible rate. We construct a code with $(r,\ell)$-cooperative locality that has  rate at least $\frac{r - \ell}{r + \ell}$. This code construction is based on the regular bipartite graphs with girth at least $\ell+1$. In the light of an upper bound  $\frac{r}{r + \ell}$ on the rate of a code  with $(r, \ell)$-cooperative locality that we show later, this construction provides codes that are very close to being optimal. Here, we also note that there are explicit constructions for the regular bipartite graphs with large girth~\cite{LazebnikC}. Thus, one can obtain high (almost optimal) rate codes with $(r, \ell)$-cooperative locality for distributed storage systems. Note that, a minimum distance is not guaranteed in this construction. We also show that the codes based on expander graphs  enable cooperative local repairs while maintaining both high rate and good minimum distance.

Given a large number of parity constraints with low weights, expander graph based codes are natural candidates for codes to enable locality. However, these codes are overkill when one is interested in code repair of single failed symbol and codes with significantly better rate vs. distance trade-off can be obtained~\cite{Gopalan12, PKLK12, RKSV12, TamoBarg}. But as we aim to recover from multiple failures in a local manner, these codes become an attractive option.
\renewcommand{\arraystretch}{1.5}
\begin{table*}[htbp]
\begin{center}
    \begin{tabular}{|p{4.1cm}| c | c | c |} 
    \hline
    Construction & \pbox{20cm}{Cooperative locality} & \pbox{20cm}{Rate ${\rm rate}(\Cc)$} & \pbox{20cm}{Minimum distance $d_{\min}(\Cc)$} \\[5pt] \hline \hline 
    Partition code (Sec.~\ref{subsec:partition})  & $(r,\ell)$& $\frac{r}{r + \ell^2}$ & $n - k + 1 - \ell\big(\frac{k\ell}{r} - 1\big)$ \\[5pt]\hline
    Product code (Sec.~\ref{subsec:product}) & $(r,\ell)$ & $\big(\frac{r}{r+1}\big)^{\ell}$ & $\ell + 1$ \\[5pt] \hline
    Concatenated code (Sec.~\ref{sec:concatenated}) & $(r,\ell)$ & $\frac{\ell + 2}{\ell + 4}\frac{r}{r + 2}$ & $\ell + 1$ \\[5pt] \hline
     \pbox{20cm}{Regular bipartite graph based \\ code (Sec.~\ref{sec:girth})}  & $(r,\ell)$&$ \geq \frac{r - \ell}{r + \ell}$ & $\geq \ell + 1$ \\[5pt] \hline
      \pbox{20cm}{Unbalanced bipartite expander \\ graph based code (Sec.~\ref{subsec:unbalanced})}  & $(r,\ell)$ & $\geq 1 + \frac{h}{\Delta}\frac{r}{\ell} - h$ & $\geq \Big(2 - \epsilon-\frac{\epsilon}{t}\Big)\alpha n$ \\[5pt] \hline 
    \pbox{20cm}{Double cover of regular expander \\ graph based code (Sec.~\ref{subsec:double_cover})} & $(r,\ell)$ & $\geq 2\frac{r}{\ell \Delta} - 1$ & $\delta(\delta - \frac{\lambda}{\Delta})n$ \\[5pt] \hline
    Hadamard code (Sec.~\ref{subsec:Hadamard})  &  \pbox{20cm}{$(r = \ell+1,\ell)$, \\ $\forall 1\le \ell \le \frac{n-1}2$} & $\frac{\log(n+1)}{n}$ & $\frac{n+1}2$ \\[5pt] \hline
     \end{tabular}
     
     \vspace{0.1in}
     \caption{Summary of the constructions of codes with $(r, \ell)$-cooperative locality considered in this paper. For the codes based on unbalanced bipartite expander graphs (Sec.~\ref{subsec:unbalanced}), we assume that the underlying bipartite graphs is bi-regular with $h$ and $\Delta$ representing its left and right degrees, respectively. Moreover, the graph exhibits expansions from left to right of any set of at most $\alpha n$ left nodes with expansion ratio $h(1- \epsilon)$. Here the constituent local codes have distance at least $t+1$. For codes based on double cover of regular expander graphs (Sec.~\ref{subsec:double_cover}), $\Delta$ and $\lambda$ denote the degree and the second largest absolute eigenvalue of the underlying graph. This construction utilizes smaller code of minimum distance at least $\delta \Delta$ to define local constraints at the vertices of the double cover.}
      \label{table:parameters}
\end{center}
\end{table*}


\subsection{Contributions and organization}
{
In Section~\ref{sec:def}, we first present a formal definition of codes with $(r, \ell)$-cooperative locality and highlight the connections between the notion of cooperative locality as defined in this paper and various other contemporary notions from distributed storage literature~\cite{LlHolOgg2013, TamoBarg, TamoBound, WangZhang, availability, PKLK12, RKSV12, KPLK12} that aim to generalize locally repairable codes (LRCs)~\cite{Gopalan12, PapDim12}. In Section~\ref{sec:disjoint_groups}, we comment on the cooperative locality parameters of the codes with multiple small sized disjoint repair groups for each code symbol~\cite{LlHolOgg2013}. In Section~\ref{sec:rdelta}, we highlight both the differences and similarities between the codes with $(r, \ell)$-cooperative locality and the codes with $(\tilde{r}, \delta)$-locality~\cite{PKLK12}.
}

In Section~\ref{sec:dmin_bound}, we obtain an upper bound on the minimum distance of a code with $(r, \ell)$-cooperative locality which encodes $k$ information symbols to $n$ symbols long codewords. 
As a special case of this result, we then obtain a bound on the   best possible rate for a code with $(r, \ell)$-cooperative locality with no further minimum
distance requirement. We address the issue of providing explicit constructions for codes with $(r, \ell)$-cooperative locality in Sections~\ref{sec:construction},~ \ref{sec:girth} and \ref{sec:expander}. 

In Section~\ref{sec:construction}, we present two simple constructions for the codes that have $(r, \ell)$-cooperative locality and comment on their rates with respect to the bound obtained in Section~\ref{sec:dmin_bound}. 
In Section~\ref{sec:girth}, we consider the codes based on regular bipartite graphs with large girth (girth = length of the smallest cycle). In particular, we show that a code based on regular bipartite graph with girth $g$ allows for cooperative local repair of $g-1$ failed code symbols. We further study cooperative locality of the codes based on expander graphs in Section~\ref{sec:expander}. We comment on the conditions in terms of expansion ratio or second eigenvalue that the underlying expander graph needs to satisfy for the code to enable cooperative repair of a certain number of erasures. Table~\ref{table:parameters} summarizes the rates and distances obtained by various code constructions considered in this paper.

Certain families of classical algebraic codes may possess local repair property. In Section~\ref{subsec:Hadamard}, we study punctured Hadamard codes (a.k.a. Simplex codes) in the context of cooperative local repair. We show that a punctured Hadamard code with $n$ symbols long codewords has $(r = \ell + 1, \ell)$-cooperative locality for any $\ell \leq \frac{n-1}{2}$. We conclude this paper in Section~\ref{sec:conclusion} with some directions for future work.

A short note on notation: we use bold lower case letters to denote vectors. For an integer $n \geq 1$, $[n]$ denotes the set $\{1, 2,\ldots, n\}$. For a code $\Cc$, we use ${\rm rate}(\Cc)$ and $d_{\min}(\Cc)$ to denote its rate and minimum distance, respectively.


\subsection{Related work}
The concept of codes with small locality for distributed storage system is introduced in~\cite{oggier_hom, Gopalan12, Papailiopoulos2012}. In \cite{Gopalan12}, Gopalan {\em et al.} study the rate vs. distance trade-off for linear codes with small locality or locally repairable codes\footnote{Throughout this paper, we use both ``codes with small locality'' and ``locally repairable codes'' to refer to the codes that enable local repair of a single failed code symbol.}. The similar trade-offs under more general definitions of locally repairable codes and constructions of the codes attaining these trade-offs are studied in \cite{PapDim12, PKLK12, RKSV12, KPLK12,cadambe_mazumdar, forbes, TamoBarg} and references therein.

In \cite{KumarTwo}, Prakash {\em et al.} consider  codes that allow for local repair of multiple code symbols. In particular, they focus on codes that can correct two erasures by utilizing two parity checks of weights at most $r + 1$. Prakash {\em et al.} derive the rate vs. distance trade-off for such code and (for large enough field size) show the existence of the codes that attain the trade-off. We note that the definition of cooperative locality considered in this paper is more general than that studied in \cite{KumarTwo}. Moreover, we do not restrict ourselves to only two erasures. In Section~\ref{subsec:kumar}, we show that the codes based on regular bipartite graphs with high girth are rate-wise (almost) optimal under the natural generalization of \cite{KumarTwo} to more than two erasures.

Recently, the codes that enables multiple ways to locally repair a code symbols have received attention. In \cite{LlHolOgg2013, TamoBarg, TamoBound}, the codes that enable multiple disjoint repair groups for every code symbol are considered. The codes that provide multiple disjoint repair group for only information symbols are studied in \cite{availability, WangZhang}. In Section~\ref{sec:disjoint_groups}, we comment on the implication of this line of work for the issue of cooperative locality.

\section{Codes with $(r, \ell)$-cooperative locality}
\label{sec:def}

\begin{definition}
\label{def:def1}
A $q$-ary code $\Cc$ with length $n$ and dimension $k\equiv \log_q |\Cc|$ is called an $(n, k)$ code. We define an $(n,k)$ code $\Cc$ to be a 
code with $(r, \ell)$-cooperative locality if for each $\Sc \subset [n]$ with $|\Sc| = \ell$, we have a  set $\Gamma_{\Sc} \subseteq [n]\backslash \Sc$ such that 
\begin{enumerate}
\item $|\Gamma_{\Sc}| \leq r$,
\item For any codeword $\cv = (c_1, c_2,\ldots, c_n) \in \Cc$, the $\ell$ code symbols $\cv_{\Sc} := \{c_i : i \in \Sc\}$ are functions of the code symbols $\cv_{\Gamma_{\Sc}}:= \{c_i : i \in \Gamma_{\Sc}\}$. 
\end{enumerate}
\end{definition} 

Note that Definition~\ref{def:def1} ensures that any $\ell$ code symbols can be {\em cooperatively} repaired from at most $r$ other code symbols. This generalizes the notion of codes with {\em all-symbol locality} $r$~\cite{Gopalan12, PapDim12, oggier_hom}, where locality is defined with respect to one code symbol, {\em i.e.}, $\ell = 1$. 

\begin{remark}
For a code $\Cc$ with all-symbol locality $r$, we have the following bound on its minimum distance~\cite{Gopalan12, PapDim12}. 
\begin{align}
\label{eq:dmin_gopalan}
d_{\min}(\Cc) \leq n - k - \ceilb{\frac{k}{r}} + 2.
\end{align}
The code attaining the bound in \eqref{eq:dmin_gopalan} are presented in \cite{PapDim12, RKSV12, KPLK12, TamoBarg} and references therein.
\end{remark}

\subsection{Cooperative locality from codes with multiple disjoint local repair groups for code symbols}
\label{sec:disjoint_groups}
In \cite{LlHolOgg2013, TamoBarg, TamoBound}, codes with multiple disjoint local repair groups for all code symbols are studied. These codes allow for multiple ways to recover a particular code symbol by contacting disjoint sets of small number of other code symbols. In particular, the work in  \cite{LlHolOgg2013, TamoBarg, TamoBound} study codes with at least $t$ disjoint local repair groups, each comprising of at most $\tilde{r}$ other code symbols. We claim, according to our definition, these codes also have $(\tilde{r}i, \ell = i)$-cooperative locality for each $i \in [t]$. Without  loss of generality, we establish this for $i = t$, {\em i.e.},  we argue that a code with $t$ disjoint repairs groups (each of size at most $\tilde{r}$) has $(\tilde{r}t, \ell = t)$-cooperative locality. 

Consider a set of $t$ code symbols in failure. For any of these $t$ failed code symbols, each symbol can have at least one failed code symbol in at most $t-1$ of its $t$ disjoint repair groups. This implies that the code symbol under consideration has at least one of its local repair groups free of any failures. Thus, the code symbol can be repaired with the help of one of its intact local repair groups. This leave us with $t-1$ code symbols in failure (erasure). Now, for another code symbol in failure, we can have at most $t-2$ of its disjoint local repair groups with at least one failed code symbol. This leaves at least $2$ of its disjoint local groups intact; therefore, this code symbol can be repaired with the help of one of its intact local repair groups. Following the similar argument, we can see that all of the $t$ failed code symbols can be repaired in a code with $t$ disjoint repair groups for all code symbols. In the worst case, we contact at most $\tilde{r}t$ code symbols to repair all of the $t$ failures. This establishes the $(\tilde{r}t, \ell = t)$-cooperative locality for the codes under consideration.

Similarly the codes with availability~\cite{WangZhang, availability}, which enable multiple disjoint repair groups only for information (systematic) symbols in a codeword, can allow for cooperative local repair for certain ranges of system parameters. In particular, \cite[Construction I]{availability} can give codes with $(\tilde{r}\ell, \ell)$-cooperative locality and rate $\frac{\tilde{r}}{\tilde{r} + \ell}$.

\begin{remark}
Here, we would like to note that the definition of the codes with $(r, \ell)$-cooperative locality (Definition~\ref{def:def1}) is more general. In particular, we show in  Section~\ref{subsec:partition}, Section~\ref{sec:concatenated} and Section~\ref{sec:girth} that it is possible to have codes with $(r, \ell)$-cooperative locality that do not have at least $t = \ell$ disjoint local repair groups for all code symbols (or information symbols).
\end{remark}

\subsection{Comparison with the codes with $(\tilde{r}, \delta)$-locality~\cite{RKSV12, KPLK12}}
\label{sec:rdelta}
In \cite{PKLK12}, Prakash {\em et al.} propose to study  codes with $(\tilde{r}, \delta)$-locality, a generalization that enforces additional requirements which the
 local repair groups of an LRC need to satisfy. In particular, a code $\Cc$ is said to have $(\tilde{r}, \delta)$-locality if there is a set of codes $\{\cC_i\}_{i \in \Lc}$ obtained by 
 puncturing the code $\Cc$, for some index set $\Lc$, such that the following three requirements hold: 1) For each $i \in \Lc$, the support of $\cC_i$ is no more than $\tilde{r}+\delta-1$, 2) for each $i \in \Lc$, the minimum distance of $\Cc_i$ is larger than or equal to $\delta$, and 3) each code symbol is contained in the support of at least one of the punctured codes $\cC_i$, $i \in \Lc$. The rate vs. distance trade-offs for the codes with $(\tilde{r}, \delta)$-locality and the constructions attaining these trade-offs are presented in \cite{RKSV12, KPLK12}.

Note that a code with $(r, \delta)$-locality ensures repair of any $\delta - 1$ failures within each punctured code. Here, we would like to highlight that the notion of $(r, \ell)$-cooperative locality is different from that of $(r, \delta)$-locality. In particular,  codes with $(r, \ell)$-cooperative locality are not required to meet the requirement 2) in the aforementioned definition of the codes with $(r, \delta)$-locality. As a result, there are families of codes which satisfy the requirements of $(r, \ell)$-cooperative locality, but that do not meet the definition of the codes with $(r, \delta)$-locality. We illustrate this with the help of the following example.

Let $\Cc$ be a code which encodes $3$ message symbols $\mv = (a, b, c)$ to a $7$ symbols long codeword $$\cv = (a, b, c, a+b, b+c, c + a, a+b+c).$$ We note that the code $\Cc$ is nothing but a $[7, 3, 4]$ Simplex code which we study in Sec.~\ref{subsec:Hadamard}. It follows from the analysis presented in Sec.~\ref{subsec:Hadamard} that this code has $(r = 3, \ell = 2)$-cooperative locality, {\em i.e.}, any set of $\ell = 2$ failed code symbols can be recovered by contacting $r = 3$ other code symbols. Let's assume that the code symbols $a$ and $a + b$ are in failure. In this case we can recover both the failed code symbols from the set of $r = 3$ code symbols $(b, b + c, a + b+ c)$. In other words $(a, a+b, b, b+c, a+b+c)$ form a punctured code of the original code $\Cc$ at $r + \ell = 5$ indices. However, this punctured code does not have minimum distance at least $\ell + 1 = \delta = 3$. This can easily be observed from the fact that the punctured sub-code does not allow the repair of $2$ code symbols $b + c$ and $a + b+ c$ from the remaining set of $3$ code symbols $(a, a+b, b)$. Moreover, there is no punctured codes of the original code at at most $r + \ell = 5$ indices which has minimum distance at least $\ell + 1 = 3$. Therefore, $\Cc$ is an example of a code with $(r = 3, \ell = 2)$-cooperative locality which does not have $(r = 3, \delta = \ell + 1= 3)$-locality as defined in \cite{PKLK12}. 

This also shows that the definition of $(r, \ell)$-cooperative locality is not a strengthening of the definition of $(r, \delta)$-locality. Hence, one cannot directly invoke the impossibility results for the codes with $(r, \delta)$-locality to obtain impossibility results for the codes with $(r, \ell)$-cooperative locality. However, as far as the achievability is concerned, a construction for a code with $(\tilde{r}, \delta)$-locality gives a construction with $(r = \ell\tilde{r}, \ell = \delta - 1)$-cooperative locality as explained in Sec.~\ref{sec:construction}.

\section{Rate vs. Distance Trade-off for Codes with $(r, \ell)$-cooperative locality}
\label{sec:dmin_bound}

In this section, for given $r$ and $\ell$, we present a trade-off between the rate and the minimum distance of a code with $(r, \ell)$-cooperative locality (cf. Definition~\ref{def:def1}). We employ the general proof technique introduced in \cite{Gopalan12, cadambe_mazumdar, forbes} to obtain the following result.

\begin{theorem}
\label{thm:bound1}
Let $\Cc \subseteq \FF_q^n$ be an $(n, k)$ code (linear, or non-linear) over the finite field $\FF_q$ with $(r, \ell)$-cooperative locality. Then, the minimum distance of $\Cc$ satisfies 
\begin{align}
\label{eq:bound1new}
d_{\min}(\Cc) \leq  n - k +1 - \ell\floorb{\frac{k -\ell}{r}}.
\end{align}
Furthermore, when we have $r \geq \ell$, the minimum distance of $\Cc$ satisfies the following.
\begin{align}
\label{eq:bound1}
d_{\min}(\Cc) \leq  n - k +1 - \ell\left(\ceilb{\frac{k}{r}} - 1\right).
\end{align}
\end{theorem}

\begin{figure}[htbp]
\algrule[1pt]
\textbf{Algorithm:} Construction of sub-code $\cC' \subset \cC$.
\algrule[1pt]
\begin{algorithmic}[1]
\REQUIRE $(n, k)$ code $\cC \subseteq \F_q^{n}$ with $(r, \ell)$-cooperative locality.
\STATE $\cC_0 = \cC$
\STATE $j = 0$  
\WHILE{$|\cC_{j}| > q^{\ell}$}
\STATE $j = j + 1$.
\STATE Choose $i^j_1,~i^j_2,\ldots, i^j_{\ell} \in [n]$ such that, for every $m \in [\ell]$, there exist at least two codewords in $\Cc_{j-1}$ that differ at $i^j_m$-th coordinate.
\STATE Let $\Rc_{j} = \Gamma_{\{i^{j}_1, \ldots, i^{j}_{\ell}\}}$ be the index of at most $r$ code symbols that cooperatively repair $\ell$ code symbols indexed by $\{i^{j}_1, \ldots, i^{j}_{\ell}\}$.
\STATE Let $\yv \in \F_q^{|\Rc_{j}|}$ be the most frequent element in the multi-set $\{\xv_{\Rc_j}: \xv \in \cC_{j-1} \subset \F_q^n\}$.
\STATE Define $\cC_{j} := \{\xv: \xv \in \cC_{j-1} \subset \F_q^n~\text{and}~\xv_{\Rc_j} = \yv\}$.
\IF{$1 < |\cC_{j}| \leq  q^{\ell}$}
\STATE \textbf{end while}
\ELSIF{$|\cC_{j}| = 1$}
\STATE Pick a maximal subset $\widetilde{\Rc}_j \subseteq \Rc_j$ such that $|\tilde{\cC_j}| > 1$, where $\widetilde{\cC}_{j} := \{\xv: \xv \in \cC_{j-1} \subset \F_q^n, \xv_{\widetilde{\Rc}_j} =\widetilde{\yv}_{j}\}$ and $\widetilde{\yv}_{\,j} \in \F_q^{|\widetilde{\Rc}_{j}|}$ be the most frequent element in the multi-set $\{\xv_{\widetilde{\Rc}_j}: \xv \in \cC_{j-1} \subset \F_q^n \}$.
\STATE $\cC_j = \widetilde{\cC}_j$. 
\STATE \textbf{end while}.
\ENDIF
\ENDWHILE
\ENSURE $\cC' = \cC_j$.
\end{algorithmic}
\algrule[1pt]
\caption{Construction of sub-code $\cC' \subset \cC$.}
\label{algo:subcode}
\end{figure}

\begin{proof}
The proof involves construction of a sub-code $\Cc' \subset \Cc \subseteq \FF_q^n$ such that all but a small number of coordinates in every codeword of $\Cc'$ are fixed. The coordinates of the codewords in $\Cc$ are fixed in an iterative manner as follows. In each iteration, we consider a set of $\ell$ coordinates which have not been fixed so far. Then we pick the set of $r$ other coordinates such that the code symbols associated with these $r$ coordinates allow us to repair the code symbols associated with the $\ell$ coordinates under consideration. The current iteration ends with fixing these $r + \ell$ coordinates to some specific values. Note that some of the $r$ coordinates may have been fixed in the previous iterations. We describe the iterative construction of the sub-code $\Cc'$ in Fig.~\ref{algo:subcode}. Given the sub-code $\mathcal{C}' \subset \mathcal{C}$, we have
\begin{align}
\label{eq:dmin1}
d_{\min}(\Cc) \leq  d_{\min}(\Cc').
\end{align}
Given $\Cc'$, one can obtain a code $\Cc''$ with $|\Cc''| = |\Cc'|$ by removing fixed coordinates from all the codeword in $\Cc'$. This implies that $d_{\min}(\Cc'') = d_{\min}(\Cc')$, which along with \eqref{eq:dmin1} give us the following. 
\begin{align}
\label{eq:dmin2}
d_{\min}(\Cc) \leq d_{\min}(\Cc'').
\end{align}
 We refer the reader to Appendix~\ref{appen:thm1} for the complete proof.
\end{proof}

\begin{remark}
It is possible to obtain a bound on the minimum distance of codes with $(r,\ell)$-cooperative locality that depends on the alphabet size, in the spirit of
\cite{cadambe_mazumdar}. Indeed, a more general version of Theorem \ref{thm:bound1} will give,
$$
k   \le \min_{t\le \min\{\lfloor\frac{n}{r+\ell}\rfloor,\lfloor\frac{k-1}{r}\rfloor\}} rt + \log_qA_q(n - t(r+\ell),d) ,
$$
where $A_q(n,d)$ is the maximum size of a $q$-ary error-correcting code of length $n$ and distance $d$. The proof of this bound is straight-forward.
\end{remark}

Note that an $(n, k)$ code with $(r, \ell)$-cooperative locality has its minimum distance at least $\ell + 1$ as it can recover from the erasure of any $\ell$ code symbols (cf. Definition~\ref{def:def1}). Combining this observation with Theorem~\ref{thm:bound1}, we obtain the following result. 
\begin{corollary}
\label{cor:rate1}
The rate of an $(n, k)$ code with $(r, \ell)$-cooperative locality is bounded as
\begin{align}
\label{eq:rate1new}
\frac{k}{n} \leq \frac{r}{r + \ell} +  \frac{1}{n}\frac{\ell^2}{r}.
\end{align}
Furthermore, for case when we have $r \geq \ell$, the rate of an $(n, k)$ code with $(r, \ell)$-cooperative locality satisfies
\begin{align}
\label{eq:rate1}
\frac{k}{n} \leq \frac{r}{r+ \ell}.
\end{align}
\end{corollary}
\begin{proof}
It follows from \eqref{eq:bound1new} and the fact $d_{\min}(\Cc) \geq \ell + 1$ that
\begin{align}
\ell + 1 \leq d_{\min}(\Cc) \leq  n - k +1 - \ell\floorb{\frac{k -\ell}{r}}. \nonumber
\end{align}
By using $\floorb{\frac{k -\ell}{r}} \geq {\frac{k -\ell}{r}} - 1$, we get 
\begin{align}
\label{eq:rate_gen}
\frac{k}{n} \leq \frac{r}{r + \ell} + \frac{1}{n}\frac{\ell^2}{r}.
\end{align}
In the case when we have $r \geq \ell$, we can combine \eqref{eq:bound1} with the observation $d_{\min}(\Cc) \geq \ell + 1$ to obtain the following.
\begin{align}
\ell + 1 \leq d_{\min}(\Cc) \leq  n - k +1 - \ell\left(\ceilb{\frac{k}{r}} - 1\right). \nonumber
\end{align}
By using $\ceilb{\frac{k}{r}} - 1 \geq \frac{k}{r} - 1$, we get
\begin{align}
\label{eq:rate_spec}
\frac{k}{n} \leq \frac{r}{r + \ell}.
\end{align}
\end{proof}

\begin{remark}
Here, we note that the assumption $r \geq \ell$ is a natural assumption as it always holds for linear codes with $(r, \ell)$-coopreative locality and dimension at least $\ell$, i.e., $k \geq \ell$. 
Note that the additional term $\frac{1}{n}\frac{\ell^2}{r}$ that we have for the case when $r < \ell$ vanishes as $n$ becomes large as compared to $\ell$.
\end{remark}

\section{Naive constructions of codes with $(r, \ell)$-cooperative locality}
\label{sec:construction}
In this section we address the issue of constructing high rate codes that have $(r, \ell)$-cooperative locality. 
In particular, we describe two simple constructions that ensure cooperative local repair for the failure
of any $\ell$ code symbols: 1) Partition code and 2) Product code. In Partition code, we partition the 
information symbols in groups of $\frac{r}{\ell}$ symbol and encode each group with an 
$(\frac{r}{\ell} + \ell, \frac{r}{\ell})$-MDS code (cf. Section~\ref{subsec:partition}). On the other hand, a product code is obtained by arranging $k = \big(\frac{r}{\ell}\big)^{\ell}$ information symbols in an $\ell$-dimensional array and then introducing parity symbols along different dimensions of the array (cf. Section~\ref{subsec:product}). 

\subsection{Partition Code}
\label{subsec:partition}
For the ease of exposition, we assume that $\ell | r$
~and $\big(\frac{r}{\ell}\big) | k$. Given $k$ information symbol over $\F_q$, a Partition code encodes these symbols into $n = k\frac{r + \ell^2}{r}$ symbols long codewords as follows: 
\begin{enumerate}
\item Partition $k$ information symbols into $p = \frac{k\ell}{r}$ groups of size $\frac{r}{\ell}$ each. 
\item Encode the symbols in each of the $p$ groups using an $(\frac{r}{\ell} + \ell, \frac{r}{\ell})$-MDS code over $\F_q$. We refer to the $\frac{r}{\ell} + \ell$ code symbols obtained by encoding $\frac{r}{\ell}$ information symbols in the $i$-th group as $i$-th local group.
\end{enumerate} 

As it is clear from the construction, Partition code has rate $\frac{k}{n} = \frac{r}{r + \ell^2}$. Moreover, a code symbol  can be recovered from any $\frac{r}{\ell}$ other code symbols from its local group. In the worst case, when $\ell$ failed code symbols belong to $\ell$ distinct local groups, we can recover all $\ell$ symbols from $\ell\frac{r}{\ell} = r$ code symbols, downloading $\frac{r}{\ell}$ symbols from each of the $\ell$ local groups containing one failed code symbol.

\begin{remark}
Note that the Partition codes presented here are special cases of codes with $(\frac{r}{\ell}, \delta = \ell + 1)$-locality as studied in \cite{KPLK12, RKSV12} (cf. Sec.~\ref{sec:rdelta}). The partition codes as described above only aim at maximizing the rate of the code. If we are also interested in achieving large minimum distance, then we can take $n$ strictly greater than $k\frac{r + \ell^2}{r}$ and attain the following relationship between the minimum distance $d_{\min}$ and the code dimension $k$~\cite{RKSV12}
\begin{align}
d_{\min}(\Cc) = n - k + 1 - \ell\left(\frac{k\ell}{r} - 1\right).
\end{align}
\end{remark}

In the above construction of the Partition codes we use an $(\frac{r}{\ell} + \ell, \frac{r}{\ell})$ MDS code to encode disjoint groups of $\frac{r}{\ell}$ message symbols. Note that the rate of this MDS code governs the rate of the overall code. One can potentially use some other code $\Cc^{\rm local}$ of minimum distance at least $\ell + 1$ to encode disjoint groups of $\frac{r}{\ell}$ message symbols. Now, we use $\mathpzc{r}(x)$, $x \in [\ell]$ to denote the number of symbols that needs to be contacted to repair $x$ erasure in one local group. For the case when an $(\frac{r}{\ell} + \ell, \frac{r}{\ell})$ MDS code is used, we have $\mathpzc{r}(x) = \frac{r}{\ell}$ for $x \in [\ell]$. Let $\mathpzc{r}^{\ast}(x)$ denote the {\em upper concave envelope} of $\mathpzc{r}(x)$ on the interval $[1, \ell] \in \mathbb{R}$. Assume that we have $p$ disjoint local groups, then a pattern of $\ell$ erasures can be represented by a vector $(l_1, l_2,\ldots, l_p)$. Here, $l_i$ denotes the number of erasures within the $i$-th local group. Note that we have $\sum_{i = 1}^{p}l_i = \ell$.

For a given local code $\Cc^{\rm local}$, one needs to access $\sum_{i = 1}^{p}\mathpzc{r}(l_i)$ number of intact code symbols to repair the erasure pattern $(l_1, l_2,\ldots, l_p)$. Now, we use concavity of $\mathpzc{r}^{\ast}(\cdot)$, the fact that $\mathpzc{r}^{\ast}(x) \geq \mathpzc{r}(x)$ for $x \in [\ell]$, and Jensen's inequality to obtain the following.
\begin{align}
\label{eq:jensen}
\sum_{i = 1}^{p}\mathpzc{r}(l_i) \leq \sum_{i = 1}^{p}\mathpzc{r}^{\ast}(l_i) & \leq p\mathpzc{r}^{\ast}\left(\frac{\sum_{i = 1}^{p}l_i}{p}\right)  = p\mathpzc{r}^{\ast}\left(\frac{\ell}{p}\right).
\end{align}

Since the rate of the Partition code is agnostic to the number of local groups, we can use the value of $p$ which can support $k$ message symbols and minimizes the R.H.S. of \eqref{eq:jensen}. This approach optimizes the value of $r$ for a given choice of $\ell$ and $\Cc^{\rm local}$.

\begin{example}
It is possible to achieve better locality parameters in Partition code than just to use copies of MDS codes.
Consider a Partition code with two blocks, each being a punctured Hadamard [7,3,4] code. From Theorem \ref{lem:hadamard}, we know that $r(x) = r^\ast(x) =x+1$, for $1\le x \le 3$ for these Hadamard codes. Hence, we have an $(14,6)$ code with $(5,3)$-cooperative locality.

On the other hand, consider a Partition code with two blocks of $(7,3)$-MDS codes. For this $(14,6)$ code, we may need to access up to $6$ symbols to repair even two symbols. Indeed, the overall code has $(6,3)$-cooperative locality.  
\end{example}

\subsection{Product Code}
\label{subsec:product}
Product codes are a well known construction  of codes in the coding theory literature. Given $k = \big(\frac{r}{\ell}\big)^{\ell}$ information symbols and $\ell|r$, we first arrange 
$k = \big(\frac{r}{\ell}\big)^{\ell}$ information symbols in an $\ell$-dimensional array with index of each dimension of the array ranging in the set $[\frac{r}{\ell}]$. These information symbols are then encoded to obtain an $n = \big(\frac{r}{\ell} + 1\big)^{\ell}$ symbols long code word. In the following we describe the encoding process for $\ell = 2$-dimensional array.  
The generalization of the encoding process for higher dimensions is straightforward.
\begin{enumerate}
\item Arrange $k = \big(\frac{r}{2}\big)^{2}$ information symbols in an $\frac{r}{2} \times \frac{r}{2}$ array.
\item For each row of the array, add a parity symbol by summing all $\frac{r}{2}$ symbols in the row and append these symbols to their respective rows.  
\item For each of the $\frac{r}{2} + 1$ columns of the updated array, add a parity by summing all $\frac{r}{2}$ symbols in the column.
\end{enumerate}

\begin{remark}
An $\ell$-dimensional product code enables $\ell$ disjoint repair groups for all code symbols. For example, every code symbol in a $2$-dimensional product code has two disjoint repair groups, associated with its row and column, respectively. Therefore, cooperative locality of product codes follows from the discussion in Section~\ref{sec:disjoint_groups}. We note that product codes along with their minimum distance have been previously been considered in \cite{TamoBarg, WangZhang} in the context of codes with small locality.
\end{remark}

We now compare the rate of Partition code and Product code with the bound in \eqref{eq:rate1}. For any $\ell \geq 1$, we have 
\begin{align}
\label{eq:prod_partition}
\left(\frac{r}{r + \ell}\right)^{\ell} \leq \frac{r}{r + \ell^2}.
\end{align}
Note that \eqref{eq:prod_partition} follows from the fact that $$r^{\ell}(r + \ell^2) \leq r(r^{\ell} + \ell r^{\ell - 1}\ell) + r\left(\sum_{i = 2}^{\ell}{\ell \choose i}r^{\ell-i}\ell^{i}\right) = r(r + \ell)^{\ell}.$$
Therefore, Partition code approach provides $(r, \ell)$-cooperative locality with a better rate. However, for all system parameters, the rate of Partition code is smaller than the known bound \eqref{eq:rate1}, {\em i.e.},
$$
\frac{r}{r + \ell^2} \leq \frac{r}{r + \ell}.
$$
Here, we would like to note that the difference between the rate achieved by the Partition code and the bound in \eqref{eq:rate1} gets smaller as the parameter $r$ becomes large as compared to the parameter $\ell$. It is an interesting problem to either tighten the bound in \eqref{eq:rate1} or present a construction for codes with $(r, \ell)$-cooperative locality which have higher rate than that of Partition code. In the next two sections we present two approaches to achieve this goal.

\section{Concatenated Codes with $(r, \ell)$-cooperative locality}
\label{sec:concatenated}

Here, we describe a family of concatenated codes with $(r, \ell)$-cooperative locality. This construction employs an MDS code and a code with small locality as inner and outer codes, respectively. In particular, we employ an $[\frac{r}{\ell} + x, \frac{r}{\ell}, x + 1]$ MDS code over $\F_q$ and an $[n_{\rm out}, k_{\rm out}]$ code with $(r_{\rm out}, \ell_{\rm out})$-cooperative locality over $\F_{q^{r/\ell}}$ as inner and outer codes, respectively. Let $\Cc$ be the concatenated code. We know that 
\begin{align}
\label{eq:rate_c1}
R = {\rm  rate}(\Cc) &= \frac{r}{r + x\ell}\cdot\frac{k_{\rm out}}{n_{\rm out}} 
\end{align}

Before we describe the concatenated codes with $(r, \ell)$-cooperative locality for general $\ell$, let's consider a few examples for small values of $\ell$.

\subsection{When $\ell = 3$}
\label{sec:ell_3}

Let us take an $[\frac{r}{3} + 1, \frac{r}{3}, 2]$ MDS code over $\F_q$ as the inner code. This code can repair any one failed code symbol by contacting the remaining $\frac{r}{3}$ code symbols. For outer code, we employ a code with $(r_{\rm out}, 1)$-cooperative locality over $\F_{q^{r/3}}$. This can repair any one super symbol (which consists of $\frac{r}{3}$ symbols of $\F_q$) by contacting $r_{\rm out}$ symbols over $\F_{q^{r/3}}$, {\em i.e.}, $r_{\rm out}\frac{r}{3}$ symbols over $\F_q$. (Note that in order to repair a super symbol, we can obtain the value of $r_{\rm out}$ required super symbols by contacting $\frac{r}{3}$ symbols over $\F_q$ from each of their corresponding codewords of the inner code.)

If $\ell = 3$ erasures lie in the inner codewords of $3$ different super symbols, then we can repair each of these erasures by contacting $\frac{r}{3}$ other code symbols. This amounts to using $3\frac{r}{3} = r$ symbols over $\F_q$. If at least $2$ erasures belong to the inner codeword of a super symbol, we can employ $(r_{\rm out}, 1)$-cooperative locality of the outer code to repair the corresponding super symbol. In the worst case, we contact $r_{\rm out}\frac{r}{3} + \frac{r}{3}$ symbols over $\F_q$, when $2$ erasures belong to one super symbol and the third erasure belongs to another super symbol. Since we want 
$$
r_{\rm out}\frac{r}{3} + \frac{r}{3} \leq r,
$$
we have $r_{\rm out} \leq 2$. Taking $r_{\rm out} = 2$, we can get the concatenated code with rate 
$$
R = \frac{r}{r + 3}\cdot\frac{r_{\rm out}}{r_{\rm out} + 1} = \frac{2r}{3(r + 3)}.
$$
Moreover, this code has minimum distance at least $4$. Now, we compare the rate of the obtained concatenated code with that of the Partition code described in Section~\ref{subsec:partition}, which has rate $\frac{r}{r + 9}$.
\begin{align}
&~\frac{2r}{3(r + 3)} > \frac{r}{r + 9} \Rightarrow~~r < 9.
\end{align}
Hence, for all $r <9$, the concatenated codes have a higher rate than the Partition codes.

\subsection{When $\ell = 4$}

Here, we focus on obtaining the codes with $(r, 4)$-cooperative locality. We use an  $[\frac{r}{4} + 2, \frac{r}{4}, 3]$ MDS code over $\F_q$ as the inner code. This code can correct $2$ erasures within an inner codeword associated with a super symbol. In order to repair a  super symbol, we employ a code with $(r_{\rm out}, 1)$-cooperative locality over $\F_{q^{r/4}}$ as an outer code. It can be easily verified that (for suitable value of $r_{\rm out}$) the concatenated code obtained by this approach allows for the recovery of $4$ erasures by contacting at most $r$ symbols over $\F_q$. In particular, when the inner codeword associated with one super symbol encounter $3$ erasures and the inner codeword associated with another super symbol encounter $1$ erasure, we contact at most $r_{\rm out}\frac{r}{4} + \frac{r}{4}$ symbols over $\F_q$. Since we need to satisfy
$$
r_{\rm out}\frac{r}{4} + \frac{r}{4} \leq r,
$$
we have $r_{\rm out} \leq 3$. Working with $r_{\rm out} = 3$, one can obtain a code with $(r, 4)$-cooperative locality and rate 
$$
R = \frac{r/4}{r/4 + 2}\cdot\frac{r_{\rm out}}{r_{\rm out} + 1} = \frac{3r}{4(r + 8)}.
$$
Moreover, the concatenated codes obtained in this manner have minimum distance at leat $2\times 3 = 6$. We obtain better rate as compared to that of the Partition codes (cf. Section~\ref{subsec:partition}), iff
\begin{align}
&~\frac{3r}{4(r + 8)} > \frac{r}{r + 16} \Rightarrow~~r < 16.
\end{align}

\subsection{General values of $\ell$}

Here,  in addition to $\ell | r$, we assume that $\ell$ is even\footnote{This assumption is just for the ease of exposition and a similar construction for odd $\ell$ can also be proposed.}. For $1 \leq x \leq \ell - 1$, We now take an $[\frac{r}{\ell} + x, \frac{r}{\ell}, x + 1]$ MDS code over $\F_q$ as the inner code. For outer code, we employ a code over $\F_q^{r/\ell}$ that can locally recover $\floorb{\frac{\ell}{x + 1}}$ failed (erased) super symbols. Note that there exist such codes with rate~(cf. Section~\ref{subsec:partition} and Section~\ref{sec:disjoint_groups}) $$\frac{\tilde{r}}{\tilde{r} + \floorb{\frac{\ell}{x + 1}}},$$ where $\tilde{r}$ denotes the number of super symbols needed for the local repair of $1$ super symbol. In our definition, these codes have $(i\tilde{r}, i)$-cooperative locality for all $i \in \{1, 2,\ldots, \floorb{\frac{\ell}{x + 1}}\}$. 

Now, consider the case where all $\ell$ erasures lie in the inner codewords corresponding to different super symbols, we can repair all $\ell$ erasures by contacting $\ell \times \frac{r}{\ell} = r$ code symbols over $\F_q$. For the case where $1\leq y \leq \floorb{\frac{\ell}{x + 1}}$ super symbols are in erasure, in the worst case, we have $(x + 1)$ erasures in the inner codewords corresponding to $y$ distinct super symbols and $1$ erasure in the inner codewords associated with $\ell - y(x+1)$ different super symbols. In order to repair these erasures, we contact $$y\tilde{r}\frac{r}{\ell} + (\ell - y(x+1))\frac{r}{\ell}$$ symbols over $\F_q$. Since we need to satisfy, 
$$
y\tilde{r}\frac{r}{\ell} + (\ell - y(x+1))\frac{r}{\ell} \leq r, 
$$ 
we get $\tilde{r} \leq x + 1$. Therefore, the rate of the concatenated code we get is 
\begin{align}
\label{eq:rate_x}
R(x) &= \frac{x + 1}{x + 1 + \floorb{\frac{\ell}{x+1}}}\cdot \frac{\frac{r}{\ell}}{\frac{r}{\ell} + x} \nonumber \\
&= \frac{x + 1}{x + 1 + \floorb{\frac{\ell}{x+1}}}\cdot \frac{r}{r + x\ell}.
\end{align}
Note that the concatenated code obtained in this section has minimum distance at least $(x+1)(\floorb{\frac{\ell}{x+1}} + 1)$.

\begin{remark}
If we substitute $x = \frac{\ell}{2}$ in \eqref{eq:rate_x}, we obtain a code with rate  
\begin{align}
\label{eq:rate_x1}
R\left(\frac{\ell}{2}\right) &= \frac{\frac{\ell}{2} + 1}{\frac{\ell}{2} + 1 + 1}\cdot \frac{r}{r + 2} = \frac{\ell + 2}{\ell + 4}\cdot \frac{r}{r + 2}.
\end{align}
The  rate $R\left(\frac{\ell}{2}\right) $ in \eqref{eq:rate_x1} is strictly greater than the rate of the partition codes $\frac{r}{r + \ell^2}$ as long as $r < \frac{\ell^3}{4}$.
\end{remark}

\section{Cooperative locally repairable codes using codes on graphs}

The concatenated codes described in Section~\ref{sec:concatenated} enable $(r, \ell)$-cooperative locality with better rate and minimum distance as compared to those of Partition codes. However, the improvements obtained by concatenated code approach are small and limited to the bounded values of the parameter $r$. In this section, we present various graphs based codes that improve upon the previously described approaches for a large range of system parameters.

\subsection{Bipartite graphs with large girth}
\label{sec:girth}

The {\em girth} of a graph is the  number of vertices in the shortest cycle of the graph.
In this section, we explore a particular class of codes based on bipartite graphs with high girth. In this construction, the code symbols are associated with the edges of a bipartite graph and both the left and right vertices in the the bipartite graph enforces the constraints on the code symbols associated with the edges incident on these vertices. The analysis of the cooperative locality of the codes obtained in this manner is based on the fact that the underlying bipartite graph has high girth. 

Let $\Gc = (\Uc \cup \Vc, \Ec)$ be a bipartite graph where $\Uc$ and $\Vc$ denote the set of left and right vertices, respectively. In particular, we work with the bipartite graphs that are bi-regular, {\em i.e.}, all the vertices from one part have the same degree. If all the left vertices and right vertices have degree $\Delta_1$ and $\Delta_2$, respectively, then we refer to such a bipartite graph as a $(\Delta_1, \Delta_2)$-regular bipartite graph. In the case, where we have $\Delta_1 = \Delta_2 =  \Delta$, we simply call the bipartite graph as $\Delta$-regular bipartite graph. Given the bipartite graph $\Gc$, we obtain a code  $\Cc$ (over $\F_q$) in the following manner:
\begin{itemize}
\item We assign each edge in the bipartite graph $\Gc(\cV,\cE)$ with a code symbol in the codewords of $\Cc$. That is, $\cC \subseteq \ff_q^{|\cE|}$.
\item For every (left or right) vertex in the bipartite graph, all the code symbols associated with the edges incident on the vertex satisfy a linear constraint (over $\F_q$).
\end{itemize}

Before stating our general result on the cooperative locality of the codes obtained in this manner, we consider small values of $\ell$. Note that any two edges in $\Gc$ (code symbols in a codeword of $\Cc$) share at most one vertex (appear together in at most one local constraint). Thus, for $\ell = 2$ code symbols in erasure, it is possible to find two local constraints that contain exactly one of the two erased symbols. This allows for the repair of both the erased symbols by utilizing these two local constraints. In other words, the code $\Cc$ has $(2(\Delta_{\rm max} - 1), 2)$-cooperative locality, where $\Delta_{\rm max}$ denotes the maximum degree of the underlying bipartite graph $\Gc$. Similarly, even in the presence of $\ell = 3$ erasures, one can find at least two local constraints such that there is only one erasure among the code symbols participating in each of these constraints. Fig.~\ref{fig:3erasures1}~and~\ref{fig:3erasures2} illustrate this fact by considering two possible patterns of $\ell = 3$ erasures. Now, using these two constraints one can repair two erasures, which leaves only one erased symbol which can then be recovered with the help of any of the two local constraints it appears in. The repair of $\ell = 3$ erasures involves at most $3(\Delta_{\rm max} - 1)$ other code symbols; hence, the code $\Cc$ has $(3(\Delta_{\rm max} - 1), 3)$-cooperative locality. In order to cooperatively repair $\ell > 3$ erasures in $\Cc$, we utilize the fact that the underlying bipartite graph $\Gc$ has high girth.

\begin{figure*}[t!]
        \centering
        \begin{subfigure}[b]{0.40\textwidth}
	\centering
		\includegraphics[width=0.7\textwidth]{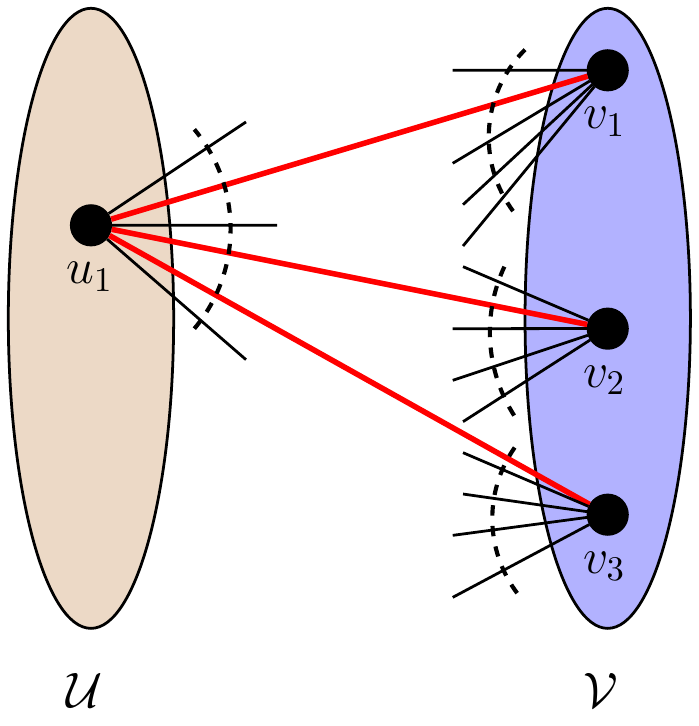}
                      \caption{In this pattern all three erased code symbols appear together in a local constraint associated with the left vertex $u_1$. Therefore, all the erased symbols must appear in different local constraints corresponding to the right vertices. This allows for the recovery of all three erased code symbols by using local constraints associated with the right vertices $v_1, v_2,$ and $v_3$.}
                       \label{fig:3erasures1}
        \end{subfigure}%
       \quad
        \begin{subfigure}[b]{0.4\textwidth}
	\centering
		\includegraphics[width=0.7\textwidth]{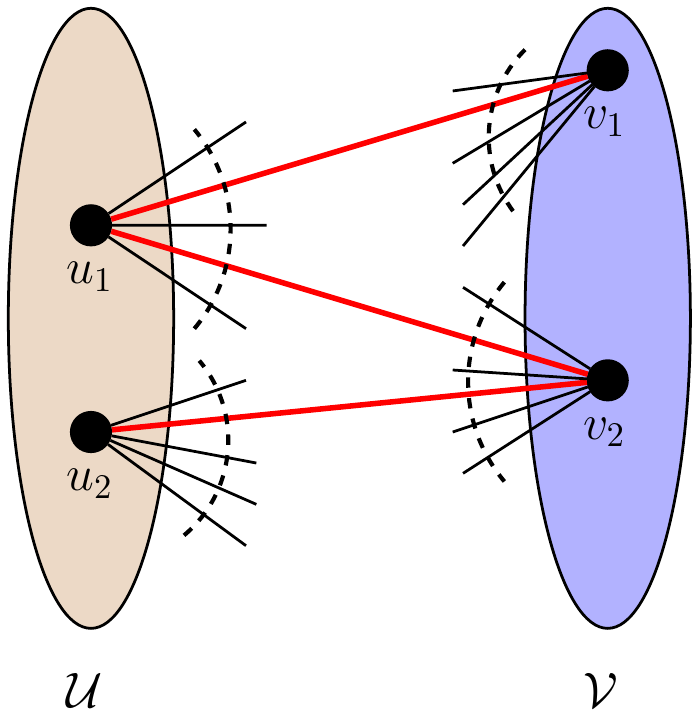}
                      \caption{Here, two erased code symbols appear together in the local constraint associated with the left vertex $u_1$. Thus, both the erased symbols appearing in two different local constraints defined by the right vertices $v_1$ and $v_2$. Since the third erased symbol can participate in only one of the local constraints associated with the right vertices ($v_2$ here), one can recover all three erased symbols.}
                      \label{fig:3erasures2}
        \end{subfigure}
        \caption{Two representative patterns of $3$ erasures. The edges associated with the three \newline erasures are colored in red.}
\end{figure*}

\begin{theorem}
\label{lem:girth_repair}
Let $\Gc$ be a $\Delta$-regular bipartite graph with girth $g$, then the code $\Cc$ obtained from the construction described above has $((g-1)(\Delta -1), g - 1)$-cooperative locality.
\end{theorem}

\begin{proof}
A bipartite graph can only have cycles of even length (number of vertices or edges).
Note that as we explain this before stating this theorem, the code $\Cc$ can correct up to $3$ erasures without any assumption on the girth of the bipartite graph $\Gc$. Therefore, without loss of generality we can assume that the girth of the bipartite graph $\Gc$ is at least six\footnote{Note that Theorem~\ref{lem:girth_repair} guarantees cooperative local repair of only $\ell = 3$ erasures when $g = 4$.}, {\em i.e.}, $g \geq 6$. We use induction over the number of erasures to prove the claim. For the base case we consider the case of $\ell = 3$ erasures. As described in the paragraph preceding the statement of this theorem, the code $\Cc$ can recover from $3$ erasures in a cooperative manner.  

Now as an inductive hypothesis, we assume that the code $\Cc$ can repair at most $\ell  \leq g - 2$ erasures in a cooperative manner and show that it is also possible to repair $\ell + 1 \leq g - 1$ erasures. Towards this, we show that given $\ell + 1$ erasures, it is possible to obtain a local constraint which has a single erasure among the code symbols appearing in the constraint. Finding such a constraint allow for the recovery of one erasure leaving $\ell$ erasures. In order to show a contradiction, we assume that no such local constraint exists. We start with a vertex say $u_1 \in \Uc$ with at least $2$ of the code symbols associated with the edges incident on it in erasure. We then traverse along one of the edges out of the vertex $u_1$ which have their corresponding code symbols in erasure. (Note that there are at least $2$ of such edges.) Let $v_1 \in \Vc$ denote the vertex that we arrive at after traversing the edge. Since $v_1$ has at least $2$ code symbols associated with its edges in erasure, we can now pick an edge associated with one of the erased symbol to reach another vertex $u_2 \in \Uc$ which is different from $u_1$. We continue this process until we can not traverse to an unexplored vertex through an edge with its associated symbol in erasure. Note that this process is bound to end in at most $\ell + 1$ steps as there are only $\ell + 1$ erasures.  This process can end with two possibilities: 1) we have traversed through all edges associated with erased symbols or 2) all the unexplored edges from the last vertex leads to previously visited vertices. The first possibility is not feasible under our assumption as it implies that the last vertex has only single erasure associated with the edges incident on it. The second possibility leads to the existence of cycle of length at most $\ell + 1$ which is infeasible as $\ell + 1 \leq g -1$. This leads to a contradiction. Thus, it is possible to obtain a local constraint which has a single erasure among the code symbols appearing in the constraint. Now that we are remained with $\ell$ erasures, we can employ the inductive hypothesis to complete the proof.

As for the total number of intact code symbols contacted during the repair process, in the worst case, we may need to utilize $g-1$ different local constraints to recover from $g-1$ erasures. This amounts to contacting $(g-1)(\Delta - 1)$ intact code symbols. 
\end{proof}

\begin{remark}{\em (Construction of regular bipartite graphs with large girth)}
\label{rem:construction_girth}
The problem of constructing regular bipartite graphs with large girth has received significant attention in the past. Here, we like to point out the work presented in \cite{LazebnikD, LazebnikC} and references therein. For an odd integer $k \geq 1$ and prime power $q$, Lazebnik {\em et al.} present explicit construction for $q$-regular bipartite graphs with girth at least $k + 5$ and number of edges $q^{k-1}$~\cite{LazebnikC}. Therefore, for any $\ell$, one can design a code using a regular bipartite graph from \cite{LazebnikC} which ensures cooperative local repair of any $\ell$ erasures.
\end{remark}

\subsubsection{Rate and distance of $\Cc$ obtained from a regular bipartite graph}
\label{subsec:rate_girth}

When $\Gc$ is a regular bipartite graph of degree $\Delta$, the number of independent linear constrains on the codewords is at most $\frac{2|\cE|}{\Delta}$. Hence the rate of the code is $${\rm rate}(\Cc) \geq \frac{|\Ec| - 2|\Ec|/\Delta}{|\Ec|}  = \frac{\Delta-2}{\Delta}.$$

Note that Theorem~\ref{lem:girth_repair} establish that the code $\Cc$ obtained using a $\Delta$-regular graph with girth $g$ has $((g-1)(\Delta - 1), g - 1)$-cooperative locality. If we set $(g-1)(\Delta - 1) = r$ and $g-1 = \ell$, then the following holds for the code $\Cc$ with $(r, \ell)$-cooperative locality. 
\begin{align}
\label{eq:rate_girth}
{\rm rate}(\Cc) \geq \frac{\frac{r}{\ell} - 1}{\frac{r}{\ell} + 1} = \frac{r - \ell}{r + \ell}.
\end{align} 
As far as the minimum distance $d_{\min}(\Cc)$ of a code $\Cc$ based on a $\Delta$-regular bipartite graph $\Gc$ with girth $g$ is concerned, we have the following trivial bound from Theorem~\ref{lem:girth_repair}.
\begin{align}
\label{eq:dist_girth1}
d_{\min}(\Cc) \geq g.
\end{align}
One can construct a Tanner graph $\Hc$ corresponding to the graph $\Gc$. The left vertices  and right vertices in this Tanner graph correspond to the edges in the graph $\Gc$ and the vertices in the graph $\Gc$, respectively. The Tanner graph $\Hc$ is a bi-regular bipartite graph with  left degree $2$ and right degree $\Delta$. Moreover, the girth of $\Hc$ is $2g$. We can now use \cite[Theorem 2]{Tanner} to conclude that
\begin{align}
\label{eq:dist_girth2}
d_{\min}(\Cc) \geq \tilde{d}_{\min}\frac{(\tilde{d}_{\min}-1)^{g/2} - 1}{\tilde{d}_{\min} - 2},
\end{align}
where $\tilde{d}_{\min}$ is the minimum distance of the smaller code associated with each vertex in the graph $\Gc$. For our case of $\tilde{d}_{\min}  = 2$,  \eqref{eq:dist_girth2} does not give us anything better than \eqref{eq:dist_girth1}.

\begin{remark}
The relationship between stopping number, the smallest number of erasures that cannot be corrected under iterative decoding, and the girth of the Tanner graph associated with a code have been previously explored in the literature~\cite{stoppingset}. As described above, we can obtain a Tanner graph $\Hc$ corresponding to the graph $\Gc$. This allows us to draw the connections between Theorem~\ref{lem:girth_repair} and the literature on stopping number.
\end{remark}

\begin{remark}
Compared to \eqref{eq:rate1}, this achievability result  has a loss of at most $\frac{\ell}{r+\ell}$ from the optimal possible rate.
\end{remark}

\subsubsection{Comparison with the work in~\cite{KumarTwo}}
\label{subsec:kumar}
Recently, Prakash {\em et al.} study codes which allow for local repair of $2$ erasures~\cite{KumarTwo}. In their model, they perform the repair of the two erasures in a {\em successive manner}, where a parity constraint of weight at most $\tilde{r} + 1$ is used to repair each of the two erasures. In \cite{KumarTwo}, Prakash {\em et al.} show that such codes have their rates upper bounded by $\frac{\tilde{r}}{\tilde{r} + 2}$.

Note that their model can be generalized to $\ell \geq 2$ erasures, and one can consider codes that enable successive local repairs from $\ell$ erasures by contacting $\ell$ parity constraints of weight at most $\tilde{r} + 1$. The codes based on bipartite graphs with high girth, as proposed in this section, fall under this setting. Taking $\tilde{r} = \frac{r}{\ell}$, their rate (cf. \eqref{eq:rate_girth}) is at least $\frac{\tilde{r} - 1}{\tilde{r} + 1}$. Since the upper bound $\frac{\tilde{r}}{\tilde{r} + 2}$ from \cite{KumarTwo} still applies to these codes, they exhibit almost optimal rate.

\subsection{Expander graphs}
\label{sec:expander}

The above analysis of the construction based on bipartite graphs fails to show a high minimum distance on top of the local repair property. However with the graphical construction it is also possible to have high distance, and hence protection against catastrophic failures. Next we show how the expansion property of graphs leads to such conclusion.

\subsubsection{Unbalanced bipartite expanders}
\label{subsec:unbalanced}

Let $\Gc = (\Uc \cup \Vc, \Ec)$ be an unbalanced left regular bipartite graph with $|\Uc| = n \geq |\Vc| = m$ and left degree $h$. We assume that the graph $\Gc$ is an expander graph where expansion happens from left nodes to right nodes. In particular, we assume that for all $\Sc \subseteq \Uc$ such that $|\Sc| \leq \ell$, we have 
\begin{align}
\label{eq:expansion}
\Gamma(\Sc) \geq (1 - \epsilon)h|\Sc|.
\end{align}
Here, $\Gamma(\Sc) \subseteq \Vc$ denotes the set of right nodes that constitute the neighborhood of the nodes in the set $\Sc$. 

We now associate a code symbol with each of the left nodes in the bipartite graph $\Gc$. For $v \in \Vc$, let $\Gamma(v) \subseteq \Uc$ denote the neighborhood of the node $v$. Consider a code $\Cc$ such that for each $v \in \Vc$, the code symbols associated with $\Gamma(v)$ constitute a codeword in a shorter MDS code $\mathcal{C}_0$ with length $\Delta = |\Gamma(v)|$ and minimum distance at least $t+1$. Note that this approach of constructing codes from unbalanced expander graphs is proposed in \cite{SipserSpielman, Tanner} and references therein.

Next, we argue that for small enough $\epsilon$ (cf. \eqref{eq:expansion}), the code $\Cc$ should be able to correct any set of at most $\ell$ erasures. Note that the locality parameter $r$ is dictated by the degrees of the right nodes in the graph $\Gc$.

\begin{theorem}
Let $\Gc$ be an unbalanced (left) expander bipartite graph as defined in \eqref{eq:expansion}. If we have $\epsilon < 1 - \frac{1}{t+1}$, then the code $\Cc$ can be locally repaired from any $\ell$ or less number of  erasures by contacting at most $ \ell \Delta\cdot{\rm rate}(\Cc_0)$ code symbols. 
\end{theorem}

\begin{proof}
We prove the claim using induction on $\ell$. Note that a single erasure can be repaired by using one of the local constraints the erased code symbol participates in. Now assume that at most $\ell - 1$ erasures can be repaired by using local constraints defined by the graph $\Gc$. We now show that any set of $\ell$ erasures can also be repaired using local constraints. 

Let $\Sc \subseteq \Uc$ with $|\Sc| \leq \ell$ denote the set of $\ell$ erased code symbols. In order to repair these $\ell$ erasures, we start with a right node which has at most $t$ of the code symbols associated with its neighborhood in erasure. These $t$ erasures can be corrected under the local constraints satisfied by the code $\Cc$. We can then utilize the inductive hypothesis to complete the proof. 

Note that what remains to be shown is that the desirable right node with at most $t$ associated erasures exists. Towards this, we assume that there is no such right node. In other words, this implies that the induced subgraph $\widehat{\Gc}$ defined by the nodes $\Sc \cup \Gamma(\Sc)$ has at least $t+1$ edges incident on every node in $\Gamma(\Sc)$ from the nodes in $\Sc$. Therefore, we have 
\begin{align}
\label{eq:expansion1}
(t+1)|\Gamma(\Sc)| &\leq  \text{number of edges in $\widehat{\Gc}$} = h|\Sc| \nonumber \\
\Rightarrow |\Gamma(\Sc)| &\leq \frac{h|\Sc|}{t + 1}.
\end{align}
However, for $\epsilon < 1 - \frac{1}{t+1}$, it follows from \eqref{eq:expansion} that 
$$
|\Gamma(\Sc)| > \frac{h|\Sc|}{t+1}.
$$
This along with \eqref{eq:expansion1} leads to a contradiction. Hence, in the presence of at most $\ell$ erasures it is possible to find the desirable right node (with at most $t$ erasures among the code symbols associated with its neighborhood). 

Now the claim that  $r \leq \ell \Delta\cdot{\rm rate}(\Cc_0)$ follows from the fact that correcting each erasure requires contacting at least $\Delta\cdot{\rm rate}(\Cc_0)$ code symbols from a codeword of the shorter code $\Cc_0$.
\end{proof}

Let $\alpha n$ be such that the graph $\Gc$ allows for expansion of all sets $\Sc \subseteq \Uc$ of size at most $\alpha n$ by a factor of at least $(1 - \epsilon)h$, {\em i.e.}\footnote{As shown above, $\alpha n \geq \ell$ is a sufficient condition for the code obtained from the bipartite graph $\Gc$ to be able to allow for cooperative repair of $\ell$ erasures.},
$$
\Gamma(\Sc) \geq (1 - \epsilon)h|\Sc|~~\text{for all}~\Sc \subseteq \Uc~\text{with}~|\Sc| \leq \alpha n.
$$
\begin{proposition}\label{prop:expander}
 For the code $\Cc$ based on the bipartite graph $\Gc$ above and local codes of minimum distance $t+1$, we have, $$d_{\min}(\Cc) \geq \left(2-\epsilon-\frac{\epsilon}{t}\right)\alpha n.$$
 \end{proposition}
A proof of this fact, which is an extension of existing results (such as \cite{SipserSpielman}) is provided in  Appendix \ref{appen:expander}.
  We  further assume that the bipartite graph $\Gc$ is bi-regular with $\Delta$ denoting its right degree, {\em  i.e., $nh = m\Delta$}. Moreover, let $\Cc_0$ represent the shorter code of length $\Delta$ used to define the code $\Cc$. Then we have, 
\begin{align}
{\rm rate}(\Cc) \geq \frac{n - m\Delta(1 -  {\rm rate}(\Cc_0))}{n} = 1 + \frac{h}{\Delta}\frac{r}{\ell} - h, \nonumber
\end{align}
where $r = \ell \Delta\cdot{\rm rate}(\Cc_0)$ denotes the maximum number of intact code symbols that need to be contacted to repair $\ell$ erasures.

\begin{remark}
\label{rem:lossless}
Here, we note that for any constant $\epsilon > 0$ and $\delta < 1$, it is possible to explicitly  construct  unbalanced expander graphs with constant degree $h$, $m = \delta n$ and expansion factor $(1 - \epsilon)h$ for $\Omega(n)$ sized subsets of left vertices~\cite{lossless}.
\end{remark}

\subsubsection{Regular expander graph}
\label{subsec:double_cover}

\begin{figure}[t!]
        \centering
                \includegraphics[width=0.3\textwidth]{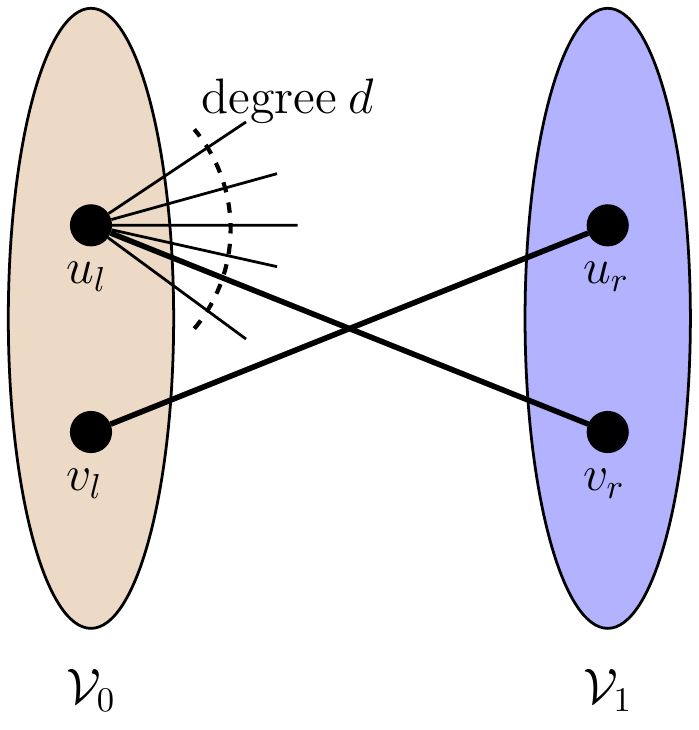}
                \caption{Illustration of the double cover $\widetilde{\Gc}$ of the $\Delta$-regular graph $\Gc$. An edge $(u, v)$ in \newline the original graph $\Gc$ contributes to two edges $(u_l, v_r)$ and $(v_l, u_r)$ in the the bipartite \newline graph $\widetilde{\Gc}$.}
\label{fig:double_cover}
\end{figure}

We now study the cooperative locality of the codes obtained by the double covers of $\Delta$-regular expander graphs~\cite{SipserSpielman}. The analysis of the cooperative locality is based on the analysis of the decoding algorithm for these codes presented in \cite{Zemor}. Note that we naturally modify the decoding algorithm from \cite{Zemor} to perform erasure correction in a cooperative manner.

Let $\Gc = (\Vc, \Ec)$ be a $\Delta$-regular graph with $|\Vc| = N$ and $\lambda$ as the second (absolute) largest eigenvalue of its adjacency matrix\footnote{If $d = \lambda_1 \geq \lambda_2 \geq \lambda_3 \geq \ldots \geq \lambda_N$ be $N$ eigenvalues of the adjacency matrix of $\Gc$, then $\lambda = \max\{\lambda_2, |\lambda_N|\}$.}. Given $\Gc$, we construct a bipartite graph $\widetilde{\Gc} = (\Vc_0 \cup \Vc_1, \widetilde{\Ec})$ with $|\Vc_0| = |\Vc_1| = N$ in the following manner (see Fig.~\ref{fig:double_cover}):
\begin{itemize}
\item Each vertex $u \in \Vc$ in the original graph $\Gc$ corresponds to a left node $u_{l} \in \Vc_0$ and a right node $u_{r} \in \Vc_1$ in the graph $\widetilde{\Gc}$.
\item For a pair of vertices $(u_l, v_r) \in \Vc_0 \times \Vc_1$, there exists an edge $(u_l, v_r) \in \widetilde{\Ec}$ iff there is an edge between the vertices $u$ and $v$ in the original graph $\Gc$, {\em i.e.}, $(u, v) \in \Ec$.
\end{itemize}

The bipartite graph $\widetilde{\Gc}$ is referred to as the double cover of the graph $\Gc$. Note that the bipartite graph $\widetilde{\Gc}$ is $\Delta$-regular with total $n = N\Delta$ edges. Moreover, the following result holds for the bipartite graph $\widetilde{\Gc}$.

\begin{lemma}{\bf (Expander Mixing Lemma)}\cite{AC88}
\label{lem:mixing}
Let $\widetilde{\Gc}$ be the $\Delta$-regular bipartite graph as described above. Then, for every $\Sc \subseteq \Vc_0$ and $\Tc \subseteq \Vc_1$, we have 
\begin{align}
\label{eq:mixing_lemma}
\left|\widetilde{\Ec}(\Sc \times \Tc) - \frac{d |\Sc| |\Tc|}{N}\right| \leq \lambda \sqrt{|\Sc| |\Tc|},  
\end{align}
where $\widetilde{\Ec}(\Sc \times \Tc)$ denotes the collection of the edges from the nodes in the set $\Sc$ to the nodes in the set $\Tc$.
\end{lemma}

Given the bipartite graph $\widetilde{\Gc}$ and a code $\Cc_{0}$ with $\Delta$-symbol long codewords, we define a code $\Cc$ as a slight generalization of the method of Sec.~\ref{sec:girth}. Each edge in $\widetilde{\Gc}$ corresponds to a code symbol in the codewords of $\Cc$. For each node in the bipartite graph $\widetilde{\Gc}$, the $\Delta$ code symbols associated with the $\Delta$ edges incident on the node constitute a codeword in the code $\Cc_{0}$. Note that we assume the local code $\Cc_0$ to be an MDS code throughout this paper.  In Fig.~\ref{fig:algo_mixing}, we present an algorithm which corrects any $\ell$ 
erasures in $\Cc$ an cooperative manner by contacting at most $\ell \Delta \cdot {\rm rate}(\Cc_{0})$ code symbols. 
The algorithm alternates between the left nodes $\Vc_0$ and the right nodes $\Vc_1$ in order to utilize the smaller code  $\Cc_0$ associated with the vertices to correct the erasures.

Let $\Sc^1 \subseteq \Vc_0$ denotes the set of nodes that have erasures among the code symbols associated with their edges and did not attempt to correct those erasures in the first round of the algorithm. This implies that each vertex in $\Sc^1$ has at least $d_{\min}(\Cc_0)$ erasures among the code symbols associated with its $\Delta$ edges. Therefore, we have 
\begin{align}
\label{eq:S0}
|\Sc^1| \leq \frac{\ell}{d_{\min}(\Cc_0)}.
\end{align}

We use $\Sc^{i},~\text{for}~i \geq 2$ to denote the set of (left or right) vertices that have erasures among the $\Delta$ code symbols associated with them in the beginning of $i$-th round and did not attempt to correct those erasures. Note that $\Sc^{i} \subseteq \Vc_0$ and $\Sc^{i} \subseteq \Vc_1$ when $i$ is an odd and even round of decoding, respectively. Next, we employ the expander mixing lemma (cf. Lemma~\ref{lem:mixing}) to show that $\big\{|\Sc^{1}|, |\Sc^{2}|, |\Sc^{3}|, \ldots \big\}$ is a strictly decreasing sequence.

\begin{figure}[t!]
\algrule[1pt]
\textbf{Algorithm:} Erasure correction in $\Cc$.
\algrule[1pt]
\begin{algorithmic}[1]
\REQUIRE A codeword from $\Cc$ with at most $\ell$ erasures.
\STATE $j = 0$.  
\WHILE{not all the erasures are corrected}
\STATE For every vertex $u \in \Vc_j$ such that $1 \leq e \leq d_{\min}(\Cc_0) - 1$ code symbols among $\Delta$ code symbols associated with it are in erasure, use erasure correcting algorithm for $\Cc_0$ to recover from these erasures. 
\STATE $j = j + 1~({\rm mod}~2)$.
\ENDWHILE
\end{algorithmic}
\algrule[1pt]
\caption{Cooperative erasure correction in the code based on the double cover of a regular expander graph.}
\label{fig:algo_mixing}
\end{figure}

\begin{lemma}
\label{lem:decoding_cover}
Let $\Sc^{1}, \Sc^{2}, \ldots$ be the sequence of sets of (left or right) vertices in the bipartite graph $\widetilde{\Gc}$ as defined above. Assume that the minimum distance of $\Cc_0$ is at least $(1 + \epsilon)\lambda$ and $\ell \leq \frac{N\lambda \epsilon \delta}{2} = \frac{n\lambda \epsilon \delta}{2\Delta}$. Then, for $i \geq 1$, we have
\begin{align}
\label{eq:Si}
\big|\Sc^{i+1}\big| \leq \frac{\big|\Sc^{i} \big|}{1+ \epsilon}.
\end{align}
\end{lemma}
\begin{proof}
We prove the relation in \eqref{eq:Si} for $i = 1$; the proof for general $i$ involves steps similar to those in the proof of the $i = 1$ case. Note that each code symbol that is in erasures after the first round of decoding is associated with some edge incident on a left node belonging to the set $\Sc^{1}$.  By the definition of the set $\Sc^{2}$, it has at least $d_{\min}(\Cc_0)$ erasures among the $\Delta$ code symbols associated with it after the first round of decoding. In other words, this implies that each vertex in the set $\Sc^{2}$ has at least $d_{\min}(\Cc_0)$ edges incident on it which are emanating from the vertices from the set $\Sc^{1}$. Therefore, we have
\begin{align}
\label{eq:S_seq}
|\Sc^{2}|d_{\min}(\Cc_0) &\leq |\widetilde{\Ec}(\Sc^{1} \times \Sc^{2})| \nonumber  \\
& \overset{(a)}{\leq}  \frac{\Delta |\Sc^{1}| |\Sc^{2}|}{N}+ \lambda \sqrt{|\Sc^1| |\Sc^2|} \nonumber \\
& \overset{(b)}{\leq}  \frac{\Delta |\Sc^1| |\Sc^2|}{N} + \lambda \frac{|\Sc^1| +  |\Sc^2|}{2} \nonumber \\
& \overset{(c)}{\leq}  \frac{\Delta \ell |\Sc^2|}{N \cdot d_{\min}(\Cc_0)}  + \lambda \frac{|\Sc^1| +  |\Sc^2|}{2},
\end{align}
where $(a)$ and $(c)$ follows from Lemma~\ref{lem:mixing} and \eqref{eq:S0}, respectively. Note that we employ the AM-GM inequality to obtain $(b)$. It follows from \eqref{eq:S_seq} that 
\begin{align}
\label{eq:S_seq1}
\big|\Sc^{2}\big| \leq \frac{\lambda}{2 \cdot d_{\min}(\Cc_0) - \lambda - 2\Delta \ell/(N \cdot d_{\min}(\Cc_0))}\big|\Sc^{1} \big|.
\end{align}
By replacing $d_{\min}(\Cc_0) = \delta \Delta$ in \eqref{eq:S_seq1}, we get
\begin{align}
\label{eq:S_seq2}
\big|\Sc^{2}\big| \leq \frac{\lambda}{2 \delta \Delta - \lambda - 2\ell/(N\delta)}\big|\Sc^{1} \big|.
\end{align}
Under our assumption that $\frac{2\ell}{N\delta} \leq \epsilon \lambda$, it follows from \eqref{eq:S_seq2} that 
\begin{align}
\label{eq:S_seq3}
\big|\Sc^{2}\big| \leq \frac{\lambda}{2 \delta \Delta - (1 + \epsilon)\lambda}\big|\Sc^{1} \big|.
\end{align}
Now, under the assumption that $\delta \Delta \geq (1 + \epsilon)\lambda$, we get from \eqref{eq:S_seq3} that
\begin{align}
\label{eq:S_seq4}
\big|\Sc^{2}\big| \leq \frac{\big|\Sc^{1} \big|}{1 + \epsilon}.
\end{align}
\end{proof}

It follows from the Lemma~\ref{lem:decoding_cover} that in at most logarithmic (in $\ell$) rounds of decoding, the algorithm described in Fig.~\ref{fig:algo_mixing} can correct $\ell$ erasures.

The codes based on the double covers of $\Delta$-regular expander graphs have been studied in the coding theory literature before (see {\em e.g.}~\cite{Zemor}). The rate and the minimum distance of the code $\Cc$ depends on the rate and the minimum distance of the code $\Cc_0$. Note that $\Cc_0$ characterizes the local constraints associated with the vertices in the bipartite graph $\widetilde{\Gc}$. In particular, if ${\rm rate}(\Cc_0)= R$ and $d_{\min}(\Cc_0) = \delta \Delta$, then we have that ${\rm rate}(\Cc) \geq 2R -1$ and $d_{\min}(\Cc) \geq \delta(\delta - \frac{\lambda}{\Delta})n$~\cite{SipserSpielman, Zemor}.

As we show in this section, for an $\epsilon > 0$ and local code $\Cc_0$ such that $d_{\min}(\Cc_0) = \delta \Delta \geq (1 + \epsilon)\lambda$, it is possible to correct $\ell \leq \frac{N\lambda \epsilon \delta}{2}$ erasures using the algorithm described in Fig.~\ref{fig:algo_mixing}. Moreover, in the worst correction of each erasure involves contacting at most 
${\rm rate}(\Cc_0)\Delta \le (1+{\rm rate}(\Cc))\Delta /2$
 other intact code symbols (assuming that the local code $\Cc_0$ is an MDS code). Therefore, the codes based on the double cover of a $\Delta$-regular expander graph and a local code $\Cc_0$ have $(r , \ell)$-cooperative locality for any $\ell \leq \frac{N\lambda \epsilon \delta}{2} = \frac{n\lambda \epsilon \delta}{2\Delta}$ and $ r = \ell(1+{\rm rate}(\Cc))\Delta /2$, that is,
 $$
 {\rm rate}(\Cc) \ge \frac{2r}{\ell \Delta} -1.
 $$

In the next section,  we show an explicit family of algebraic codes that exhibit very strong cooperative local repair property, as 
well as a very high minimum distance.

\section{Cooperative Local Repair for Hadamard Codes}
\label{subsec:Hadamard}

In this section, we study the cooperative locality for punctured Hadamard codes. Punctured Hadamard codes are also referred to as Simplex codes, which are the dual codes of Hamming codes. These codes are well known to be locally decodable codes (LDCs)~\cite{YekhaninLDC} and have multiple disjoint repair groups for each code symbols. Here, we comment on the exact parameters for the cooperative locality of these codes. In particular, we show that an $[n = 2^{k} - 1, k, 2^{k - 1}]_2$ punctured Hadamard code has $(r = \ell + 1, \ell)$-cooperative locality for any $\ell \leq \frac{n-1}{2}$. 

An $[n = 2^{k} - 1, k, 2^{k - 1}]_2$ punctured Hadamard code encodes a $k$ bits long message $(m_1, m_2,\ldots, m_k)$ to an $n = 2^{k}-1$ codeword $\cv  = (c_1, c_2,\ldots, c_{n = 2^{k}-1})$ such that 
$$
c_i  = \sum_{j = 1}^{k}m_jb^i_j~({\rm mod})~2.
$$
Here $\bv^i = (b^{i}_1, b^{i}_2,\ldots, b^{i}_k) \in \mathbb{F}_2^k$ denotes the binary representation of the integer $i \in [2^{k}-1]$. In an $[n = 2^{k} - 1, k, 2^{k-1}]_2$ punctured Hadamard code, we have $c_i + c_{2^j} = c_{i + 2^j}$, where $1 \leq j \leq k-1$~and~$i \in [2^j-1]$. Moreover, we note that an $[n = 2^{k} - 1, k, 2^{k - 1}]_2$ punctured Hadamard code has a particular structural property: for any $2 \leq \widetilde{k} < k$,  the prefix of length $2^{\widetilde{k}} - 1$ of each codeword is a codeword of the $[\widetilde{n} = 2^{\widetilde{k}} - 1, \widetilde{k}, 2^{\widetilde{k} - 1}]_2$ punctured Hadamard code which encodes the message $(m_1, m_2,\ldots, m_{\widetilde{k}})$. We now present the main result of this subsection:

\begin{theorem}
\label{lem:hadamard}
In an $[n = 2^{k} - 1, k, 2^{k - 1}]_2$ punctured Hadamard code, any $1 \leq \ell \leq \frac{n-1}{2}$ erasures can be corrected by contacting at most $\ell + 1$ other code symbols. 
\end{theorem}
\begin{proof}
We prove the Theorem by using induction over $k$. For base case, we consider $k = 2$, where the $[n = 3 = 2^2 - 1, 2, 2]_2$ punctured Hadamard code encodes the message $(m_1, m_2)$ to a codeword $(c_1, c_2, c_3) = (m_1, m_2, m_1 + m_2)$. In this case any $1 \leq \ell \leq \frac{3-1}{2} = 1$ erasure can be recovered by contacting other $\ell + 1 = 2$ code symbols. For example, one can recover $c_2 = m_2$ from $(c_1, c_3) = (m_1, m_1 + m_3)$.

For inductive step, we assume that the Lemma holds for any punctured code of dimension up to $k-1$. Consider the $[n = 2^{k} - 1, k, 2^{k - 1}]_2$ punctured Hadamard code of dimension $k$, and two cases regarding the positions of $\ell$ erased code symbols.

\begin{itemize}
\item \textbf{Case $1$:  There are $x \leq 2^{k-2} - 1$ erasures among the first $\widehat{n} = 2^{k-1} - 1$ code symbols.} Note that the first $\widehat{n} = 2^{k-1} - 1$ code symbols constitute a codeword of an $[\widehat{n} = 2^{k-1} - 1, k - 1, 2^{k-2}]_2$ punctured Hadamard code. Therefore, from the inductive hypothesis, one can correct the $x$ erasures among the first $\widehat{n}$ code symbols by contacting $x + 1$ other code symbols out of these $\widehat{n}$ code symbols. Now, if the symbol $c_{2^{k-1}}$ in erasure, we can recover it by contacting one of the intact symbol among $\{c_{2^{k-1} + 1}, c_{2^{k - 1} + 2},\ldots, c_{n = 2^{k}-1}\}$ say $c_{2^{k-1} + j}$ and the corresponding code symbol $c_{j}$ from the first $\widehat{n}$ code symbols. Now, we can repair the remaining erased symbols among $\{c_{2^{k-1} + 1}, c_{2^{k - 1} + 2},\ldots, c_{n = 2^{k}-1}\}$ from $c_{2^{k-1}}$ and the corresponding code symbol among the first $\widehat{n}$ code symbols. For example, if we want to recover the symbol $c_{2^{k-1} + m}$, we can use $c_{2^{k-1}}$ and $c_{m}$ to reconstruct $c_{2^{k-1} + m}$. In the worst case, we contact $\ell + 1$ code symbols during the repair of all $\ell$ erasures. 

\item \textbf{Case $2$: There are $x \geq 2^{k-2}$ erasures among the first $\widehat{n} = 2^{k-1} - 1$ code symbols.} 
In this case, we first recover the code symbol $c_{2^{k-1}}$, if it is in erasure. Without loss of generality we assume that $c_{2^{k-1}}$ is in erasure. Note that there are $\frac{n-1}{2} = 2^{k-1}$ distinct pairs of code symbols $\{c_i, c_{2^{k-1} + i}\}_{i \in [2^{k-1}]}$ that can recover $c_{2^{k-1}}$. Since we have at most $\frac{n-1}{2} - 1 = 2^{k-1} - 1$ erasures apart from $c_{2^{k-1}}$, one of the $2^{k-1}$ pairs $\{c_i, c_{2^{k-1} + i}\}_{i \in [2^{k-1}]}$ must be intact. This pair allows us to recover $c_{2^{k-1}}$.

Now that we know the symbol $c_{2^{k-1}} = m_k$, we can remove the contribution of $m_k$ from any of the last $2^{k} - 1 - 2^{k}$ code symbols $\{c_{2^{k-1} + 1}, c_{2^{k - 1} + 2},\ldots, c_{n = 2^{k}-1}\}$. Similarly, we can add $m_k$ to any of the first $\widehat{n} = 2^{k-1} - 1$ code symbols $\{c_1, c_2,\ldots, c_{2^{k-1}}\}$. Therefore, we can reduce the Case $2$ to Case $1$ of the proof, and repair any $\ell_1$ erasures by contacting at most $\ell_1 + 1$ code symbols.  
\end{itemize}

Combining both cases completes the proof.
\end{proof}

\begin{remark}
{Note that, for each symbol, the punctured Hadamard code provides $\frac{n-1}{2}$ disjoint repair groups. Moreover, each of these repair groups comprises $2$ symbols. Therefore, it easily follows from the discussion of Section~\ref{sec:disjoint_groups} that the punctured Hadamard code has $(2\ell, \ell)$-cooperative locality for $\ell \leq \frac{n-1}{2}$. Here, we show that these codes allow for more efficient cooperative local repair mechanism by establishing $(\ell+1, \ell)$-cooperative locality for them.}
\end{remark}

\section{Conclusion: comment of random erasures}
\label{sec:conclusion}
All the  constructions of this paper are designed to allow for the cooperative local repairs in the case of adversarial erasure patterns. One can consider the setting where erasures occur according to a random model. Here, we briefly comment on the setting where $\ell$ erasures are uniformly distributed among the code symbols. Moreover, we assume $r$ and $\ell$ to be large enough. In that case, we claim that even the simple Partition codes of Section~\ref{subsec:partition}
are asymptotically optimal.
This is true, because with reasonably high probability (depending on $r$ and $\ell$), every local group (a total $p$ of them) experiences less than about $ t \equiv \Theta\Big(\frac{\ell}{p} \log \frac{\ell}{p} \Big) $ number of erasures. Therefore, with high probability, one can perform cooperative local repair of $\ell$ random erasures even if an $\big(\frac{r}{\ell} + t, \frac{r}{\ell})$-MDS code in employed in the construction of the Partition code (cf. Sec.~\ref{subsec:partition}). This translates to a coding scheme with the overall rate of $\frac{r}{r + \ell t}$. One can take $p$ large enough to optimize this value. 
Indeed, it is possible to attain a rate of $\frac{r}{r + \ell^{1+\epsilon}}$ for some $\epsilon >0$. Comparing with \eqref{eq:rate1}, we see that Partition codes
are near-optimal in this case.
Here note that, it was shown in 
\cite{mazumdar2013local} that  for a random erasure channel, the Partition codes are asymptotically optimal in terms of achieving capacity.

\section*{Acknowledgements}
A. ~Mazumdar's research in this paper is supported by NSF CAREER grant  CCF 1453121 and grant CCF1318093. S.~Vishwanath would like to acknowledge support from Army Research Office under grant W911NF1110258. 

\bibliographystyle{unsrt}
\bibliography{cooperative_local}

\appendices

\section{Part of the Proof of Theorem~\ref{thm:bound1}}
\label{appen:thm1}

Before proceeding with the analysis, we argue the correctness of the algorithm in Fig.~\ref{algo:subcode}. Note that it is always possible to find $\ell$ coordinates $\{i^j_1, i^j_2,\ldots, i^{j}_{\ell}\}$ at line $5$. When the algorithm reaches line $5$, the sub-code $\Cc_{j-1}$ has more than $q^{\ell}$ codewords. Therefore, there must be at least $\ell$ coordinates in the codewords in $\Cc_{j-1}$ that are not fixed in the previous iterations. This also implies that, for $m \in [\ell]$, 
\begin{align}
\label{eq:Ic}
i^{j}_m \notin \Ic_{j - 1} := \bigcup_{j' \in [j-1]}\left(\Rc_{j'} \cup \{i^{j'}_1, \ldots, i^{j'}_{\ell}\}\right) \subset [n].
\end{align}
Note that the code symbols indexed by $\Ic_{j-1}$ are fixed in $\Cc_{j-1}$. This further implies that, 
$$\Rc_j = \Gamma_{\{i^j_1,\ldots, i^j_{\ell}\}} \not\subset \Ic_{j-1},$$ {\em i.e.}, not all of the code symbols contacted to repair the $\ell$ symbols indexed by the set $\{i^j_1,\ldots, i^j_{\ell}\}$ can be fixed in the previous iterations. Otherwise, the symbols indexed by the set $\{i^j_1,\ldots, i^j_{\ell}\}$ would also have been fixed in the previous iterations. 

For the construction of a sub-code as described in Fig.~\ref{algo:subcode}, we define $\Ac_{j} = \Ic_{j} \backslash \Ic_{j - 1} \subseteq \Rc_j \cup \{i^j_1, \ldots, i^j_{\ell}\}$ and $a_j = |\Ac_{j}|$. Assuming that the while loop in Fig.~\ref{algo:subcode} ends with $j = t$, for $j \in [t]$, we have 
$$
\Ic_{j} = \bigcup_{j' \in [j]}\Ac_{j'},
$$
where we take union of the disjoint sets $\Ac_{j'},~j' \in [j].$

Note that none of the indices in the set $\{i^j_1,\ldots, i^j_{\ell}\}$ corresponds to fixed symbols. Thus, by the  definition of $\Ac_j$ and $a_j$, only $a_j - \ell$ code symbols among the code symbols indexed by the set $\Rc_j$ are not fixed in the previous iterations. Hence, at line $7$, there are at most $q^{a_j - \ell}$ possibilities for $\yv_{\,j}$. This implies that
\begin{align}
\label{eq:I_2}
|\cC_j| \geq |\cC_{j-1}|/q^{a_j - \ell}.
\end{align}

The construction of the subcode $\Cc'$ can end at either line $10$ or line $14$. Here we analyze only the case when the construction ends at line $10$. (The similar analysis holds for the other case as well). In this case, we have $|\cC_{t}| \leq  q^{\ell}$, or
\begin{align}
\ell \geq \log_{q}|\cC_{t}|&\geq k - \sum_{j = 0}^{t -1}\left(a_{j+1} - \ell \right).
\end{align}
Now, using that $a_j \leq |\Ac_j| \leq |\Rc_j\cup \{i^j_1,\ldots, i^{j}_{\ell}\}| \leq r + \ell$, we get 
\begin{align}
\label{eq:mid_step}
k - \ell \leq \sum_{j = 0}^{t -1}\left(a_{j+1} - \ell \right) \leq tr.
\end{align}
This implies that 
\begin{align}
\label{eq:no_iter_new}
t \geq \floorb{\frac{k - \ell}{r}}.
\end{align}
Furthermore, for the case when we have $r \geq \ell$, it follows from \eqref{eq:mid_step} that
\begin{align}
\label{eq:no_iter}
t \geq  \ceilb{\frac{k}{r}} - 1. 
\end{align}
Note that sub-code $\cC' = \cC_{t}$. Therefore, 
\begin{align}
\label{eq:new_dim}
\log_q|\cC'| & = \log_q|\cC_{t}| \nonumber \\
&\geq \log_q|\cC| - \sum_{j = 0}^{t-1}\left(a_{j+1} - \ell\right) \nonumber \\
& = k - \sum_{j = 0}^{t-1}a_{j+1} + t\ell \nonumber \\
&\overset{(a)}{=} k - |\Ic_{t}| + t\ell
\end{align}
where $(a)$ follows from the fact that $\Ic_{t}$ is union of the disjoint sets $\Ac_j$.

Now, we define $\cC'' = \cC'|_{\Ic_{t}}$ which denotes the code obtained by puncturing the codewords in $\Cc'$ at the coordinates associated with the set $\Ic_t$. We have $|\cC''| = |\cC'|$ and $d_{\min}(\cC'') = d_{\min}(\cC')$. Moreover, the length of the codewords in $\cC''$ is $n - |\Ic_{t}|$. Next, applying the Singleton bound on $\cC''$ gives us
\begin{align}
\label{eq:dmin_mid}
d_{\min}(C) \leq d_{\min}(\cC'') &\leq n - |\Ic_{t}| - \log_q|\cC''| + 1 \nonumber \\
&\leq n - |\Ic_{t}| - (k - |\Ic_{t}| +t\ell) + 1\nonumber \\
& = n - k - t\ell + 1,
\end{align}

It follows from \eqref{eq:dmin_mid} and \eqref{eq:no_iter_new} that
\begin{align}
d_{\min}(\cC) \leq n - k +1 - \ell\floorb{\frac{k -\ell}{r}}.
\end{align}

For the setting where we have $r \geq \ell$, we can use \eqref{eq:dmin_mid} along with \eqref{eq:no_iter} to obtain that 
\begin{align}
d_{\min}(\cC) \leq n - k +1 - \ell\left(\ceilb{\frac{k}{r}} - 1\right).
\end{align}

This completes the proof.

\section{Proof of Prop.~\ref{prop:expander}}
\label{appen:expander}

Define $U_t(\Ac), \Ac\subset \Uc$, to be the set of neighbors of $\Ac$ such that each vertex of $U_t(\Ac)$ 
is connected to at most $t$ vertices from $\Ac$. Notice that for any $\Ac : |\Ac|  \le \alpha n$, $\Gamma(\Ac)  \ge (1-\epsilon)h|\Ac|$.
Furthermore,
$$
|U_t(\Ac)| + |\Gamma(\Ac) \setminus U_t(\Ac)| (t+1) \le h|\Ac|.
$$
Therefore, $|U_t(\Ac)| \ge (1-\epsilon -\epsilon/t) h|\Ac|$.

For any codeword of $\cC$ whose support is given by the vertex-set $\Sc \subset \Uc$, we must have
$U_t(\cS) = \emptyset$. Clearly, when $|\cS| \le \alpha n$, $|U_t(\cS)| \ge (1-\epsilon -\epsilon/t) h|\Sc|  >0$.
Let us assume,  $|\cS| > \alpha n$ but $|\cS| \le (2-\epsilon- \epsilon/t) \alpha n$. Let $\cQ$ be a proper subset of $\cS$ such that $
|\cQ| = \alpha n$. The number of edges coming out of $\cS \setminus \cQ$ is $h (|\cS| - \alpha n) \le h(1-\epsilon -\epsilon/t) \alpha n$.
On the other hand, $U_t(\cQ) \ge (1-\epsilon -\epsilon/t) h \alpha n$. Hence $U_t(\cS) \ne \emptyset$.

This proves that the minimum distance of the expander code is at least $(2-\epsilon- \epsilon/t) \alpha n$.

\end{document}